\documentclass[bib]{ba}
\usepackage{bayes}

\usepackage{graphicx} 
\usepackage{subfigure}
\usepackage{algorithm}
\usepackage{hyperref}
\usepackage{amsmath, amssymb, amsthm}

\usepackage{booktabs}
\usepackage{color}
\usepackage{algpseudocode}

\newcommand{\beq}{\vspace{0mm}\begin{equation}}
\newcommand{\eeq}{\vspace{0mm}\end{equation}}
\newcommand{\beqs}{\vspace{0mm}\begin{eqnarray}}
\newcommand{\eeqs}{\vspace{0mm}\end{eqnarray}}
\newcommand{\barr}{\begin{array}}
\newcommand{\earr}{\end{array}}

\newcommand{\Nmat}[0]{{{\bf N}}}

\newcommand{\bv}[0]{{\boldsymbol{b}}}

\newcommand{\lv}[0]{{\boldsymbol{\ell}}}
\newcommand{\mv}[0]{{\boldsymbol{m}}}
\newcommand{\nv}[0]{{\boldsymbol{n}}}

\newcommand{\pv}[0]{{\boldsymbol{p}}}

\newcommand{\rv}{\boldsymbol{r}}

\newcommand{\tv}[0]{{\boldsymbol{t}}}

\newcommand{\xv}{\boldsymbol{x}}
\newcommand{\yv}{\boldsymbol{y}}
\newcommand{\zv}{\boldsymbol{z}}
\newcommand{\cdotv}{\boldsymbol{\cdot}}

\newcommand{\Thetamat}{\boldsymbol{\Theta}}

\newcommand{\Phimat}{\boldsymbol{\Phi}}

\newcommand{\ellv}[0]{{\boldsymbol{\ell}}}
\newcommand{\thetav}{\boldsymbol{\theta}}

\newcommand{\phiv}{\boldsymbol{\phi}}

\newcommand{\E}{\mathbb{E}}

\newtheorem{thm}{Theorem} 

\newtheorem{prop}[thm]{Proposition}

\usepackage{array}
\newcolumntype{L}[1]{>{\raggedright\let\newline\\\arraybackslash\hspace{0pt}}m{#1}}
\newcolumntype{C}[1]{>{\centering\let\newline\\\arraybackslash\hspace{0pt}}m{#1}}
\newcolumntype{R}[1]{>{\raggedleft\let\newline\\\arraybackslash\hspace{0pt}}m{#1}}

\begin{document}

\inserttype[ba0001]{article}
\renewcommand{\thefootnote}{\fnsymbol{footnote}}
\author{M. Zhou }{
 \fnms{Mingyuan}
 \snm{Zhou}\footnotemark[1]\ead{mingyuan.zhou@mccombs.utexas.edu}

}

\title[Negative Binomial Factor Analysis ]{Nonparametric Bayesian 
Negative Binomial Factor Analysis }

\maketitle

\footnotetext[1]{
McCombs School of Business, 
The University of Texas at Austin, Austin, TX 78712, USA,
 \href{mailto:mingyuan.zhou@mccombs.utexas.edu}{mingyuan.zhou@mccombs.utexas.edu}
}
\renewcommand{\thefootnote}{\arabic{footnote}}

\begin{abstract}
A common approach to analyze a covariate-sample count matrix, an element of which represents how many times a covariate appears in a sample, is to factorize it under the Poisson likelihood. We show its limitation in capturing the tendency for a covariate present in a sample to both repeat itself and excite related ones. To address this limitation, we construct negative binomial factor analysis (NBFA) to factorize the matrix under the negative binomial likelihood, and relate it to a Dirichlet-multinomial distribution based mixed-membership model. To support countably infinite factors, we propose the hierarchical gamma-negative binomial process. By exploiting newly proved connections between discrete distributions, we construct two blocked and a collapsed Gibbs sampler that all adaptively truncate their number of factors, and demonstrate that the blocked Gibbs sampler developed under a compound Poisson representation converges fast and has low computational complexity. Example results show that NBFA has a distinct mechanism in adjusting its number of inferred factors according to the sample lengths, and provides clear advantages in parsimonious representation, predictive power, and computational complexity over previously proposed discrete latent variable models, which either completely ignore burstiness, or model only the burstiness of the covariates but not that of the factors. 

\end{abstract}

\section{Introduction}\label{sec:intro}

The need to analyze a covariate-sample count matrix, 
each of whose elements counts the number of time that a covariate appears in a sample, 
 arises in many different settings, such as text analysis, next-generation sequencing, medical records mining, and consumer behavior studies. The mixed-membership model,
 widely used for text analysis \citep{
 LDA} and population genetics \citep{pritchard2000inference}, treats each sample as a bag of indices (words), and associates each index with both a covariate  that is observed 
and a subpopulation that is latent. 
It makes the assumption that there are $K$ latent subpopulations, each of which is characterized by 
how frequent the 
covariates are relative to each other within it. Given the total number of indices for a sample, 
it assigns each index 
 independently to one of the $K$ subpopulations, with a probability proportional to the product of the corresponding covariate's relative frequency in that subpopulation and that subpopulation's relative frequency in the sample. 
A mixed-membership model constructed in this manner, as shown in \citet{BNBP_PFA_AISTATS2012} and \citet{NBP2012}, can also be connected to Poisson factor analysis (PFA) that factorizes the covariate-sample count matrix, under the Poisson likelihood, into the product 
of a nonnegative covariate-subpopulation factor loading matrix and a nonnegative subpopulation-sample factor score matrix. Each column of the factor loading matrix encodes the relative frequencies of the covariates in a subpopulation, while that of the factor score matrix encodes the weights of the subpopulations in a sample.


Despite the popularity of both approaches in analyzing the covariate-sample count matrix, 
 they both make restrictive assumptions. 
Given the relative frequencies of the covariates in subpopulations and the  relative frequencies of the subpopulations in samples, the mixed-membership model independently assigns each index to both a covariate and a subpopulation, 
and hence 
may not sufficiently capture the tendency for an index to excite the other ones in the same sample to take the same or related covariates. 
Whereas for PFA, given the factor loading and score matrices, it 
assumes that
the variance and mean are the same for each covariate-sample count, and hence is likely to underestimate the variability of overdispersed counts. 

 In practice, however, highly overdispersed covariate-sample counts are frequently observed due to self- and cross-excitation of covariate frequencies, that is to say, some covariates are particularly intense and also make other related covariates intense. For example,
the tendency for a word present in a document to appear repeatedly is a fundamental phenomenon in natural language that is commonly referred to as word burstiness \citep{church1995poisson,madsen2005modeling,doyle2009accounting}. If a word is bursty in a document, it is also common for it to excite (stimulate) related words to exhibit burstiness. 
Without capturing the self- and cross-excitation (stimulation) of covariate frequencies or better modeling the overdispersed covariate-sample counts,
the ultimate potential of the mixed-membership model and PFA will be limited no matter how the priors on latent parameters are adjusted. In addition, it could be a waste of computation if the model tries to increase the model capacity to better capture overdispersions that could be simply explained with self- and cross-excitations. 


%


 To remove these restrictions, we introduce negative binomial factor analysis (NBFA) to factorize the covariate-sample count matrix, in which we replace the Poisson distributions on which PFA is built, with the negative binomial (NB) distributions. As PFA is closely related to the canonical mixed-membership model built on the multinomial distribution, we show that 
NBFA is closely related to a Dirichlet-multinomial 
 mixed-membership (DMMM) model that uses the Dirichlet-categorical (Dirichlet-multinomial) rather than categorical (multinomial) distributions to assign an index to both a covariate and  a factor (subpopulation). 
From the viewpoint of count modeling, NBFA improves PFA by better modeling overdispersed counts, while from that of mixed-membership modeling, 
it improves the 
canonical mixed-membership model by capturing the burstiness at both the covariate and factor levels ($i.e.$, for topic modeling, it exhibits a rich-get-richer phenomenon at both the word and topic levels). In addition, we will show NBFA 
could significantly reduce the computation spent on large covariate-sample counts. 


Note that with a different likelihood for counts and a different mechanism to generate both the covariate and factor indices, NBFA and the related DMMM model proposed in the paper complement, rather than compete with, PFA \citep{BNBP_PFA_AISTATS2012,NBP2012}. 
First, NBFA provides more significant advantages in modeling longer samples, 
where there is more need to capture both self- and cross-excitation of covariate frequencies. 
Second, various extensions built on PFA or the multinomial mixed-membership model, such as 
capturing the correlations between factors \citep{CTM,DILN_BA} 
and learning multilayer deep representations \citep{ranganath2014deep,Gan2015DeepPFA,GBN}, could also be applied to 
extend NBFA. 
In this paper, we will focus on constructing a nonparametric Bayesian NBFA 
with a potentially infinite number of factors, and leave a wide variety of potential extensions under this new 
framework to future research.


To avoid the need of selecting the number of factors 
$K$, 
for PFA and the closely related multinomial mixed-membership model, a number of nonparametric Bayesian priors can be employed to support a potentially infinite number of latent factors, 
such as
 the hierarchical Dirichlet process \citep{HDP} and 
 beta-negative binomial process \citep{BNBP_PFA_AISTATS2012,NBPJordan,NBP2012}. 
To support countably infinite factors for NBFA, generalizing the gamma-negative binomial process (GNBP) \citep{NBP2012,NBP_CountMatrix}, 
we introduce a new nonparametric Bayesian prior: the hierarchical gamma-negative binomial process (hGNBP), where each of the $J$ samples is assigned with a sample-specific GNBP and a globally shared gamma process is mixed with all the $J$ GNBPs. We 
derive both blocked and collapsed Gibbs sampling for the hGNBP-NBFA, 
with the number of factors automatically inferred. 

%
 
 The remainder of the paper is organized as follows. 
 In Section \ref{sec:review}, we review 
 PFA and the multinomial mixed-membership model. 
In Section \ref{sec:NBFA}, 
we introduce NBFA and its representation as a DMMM model, 
and compare them with related models. In Section \ref{sec:hGNBP}, we propose nonparametric-Bayesian NBFA. In Section \ref{sec:inference}, we 
derive both
blocked and collapsed Gibbs sampling algorithms. 
In Section \ref{sec:results}, we first make comparisons between
 different sampling strategies and then compare the performance of various algorithms on real data. 
We 
 conclude the paper in Section 7. The proofs and Gibbs sampling update equations are provided in the Supplementary Material.

\section{Poisson factor analysis and 
mixed-membership model}\label{sec:review}

\subsection{Poisson factor analysis} \label{sec:2.2}

Let $\Nmat$ be a $V\times J$ covariate-sample count matrix for $V$ covariates among $J$ samples, where $n_{vj}$ is the $(v,j)$ element of $\Nmat$ and gives the number of times that sample $j$ has covariate $v$.
PFA factorizes $\Nmat$ under the Poisson likelihood as
\beq
\Nmat\sim\mbox{Pois}(\Phimat\Thetamat),\label{eq:PFA0}
\eeq 
where ``Pois'' is the abbreviation for ``Poisson,'' $\Phimat=(\phiv_1,\ldots,\phiv_K)\in\mathbb{R}_+^{V\times K}$ represents the factor loading matrix, $\Thetamat=(\thetav_1,\ldots,\thetav_J)\in\mathbb{R}_+^{K\times J}$ represents the factor score matrix, and $\mathbb{R}_+=\{x:x\ge 0\}$, with $\phiv_k=(\phi_{1k},\ldots,\phi_{Vk})^T$ encoding the
 weights of the $V$ covariates in factor~$k$ and $\thetav_j=(\theta_{1j},\ldots,\theta_{Kj})^T$ encoding the popularity of the $K$ factors in sample $j$. PFA 
 in \eqref{eq:PFA0} 
has an augmented representation as
\beq
n_{vj}=\sum_{k=1}^K n_{vjk},~n_{vjk}\sim\mbox{Pois}(\phi_{vk}\theta_{kj}).\label{eq:PFA1}
\eeq
As in \cite{BNBP_PFA_AISTATS2012}, it can also be equivalently constructed by 
first generating $n_{vj}$ and then assigning them to the latent factors using the multinomial distributions~as 
\beq\small
 (n_{vj1},\ldots,n_{vjK})\,|\,n_{vj}\! \sim\!\mbox{Mult}\!\left(\!n_{vj},\!\frac{\phi_{v1}\theta_{1j}}{\sum_{k=1}^K\!\phi_{vk}\theta_{kj}},\ldots,\frac{\phi_{vK}\theta_{Kj}}{\sum_{k=1}^K\!\phi_{vk}\theta_{kj}}\!\right),~n_{vj}\!\sim\!\mbox{Pois}\!\left( \sum_{k=1}^K\!\phi_{vk}\theta_{kj}\!\right).\label{eq:PFA2}
 \eeq 

\subsection{Multinomial mixed-membership model}

This alternative representation suggests a potential link of PFA to a standard 
mixed-membership model for text analysis such as probabilistic latent semantic analysis (PLSA) \citep{Hofmann99probabilisticlatent} and latent Dirichlet allocation (LDA) \citep{LDA}. 
Given the factors $\phiv_k$ and factor proportions $\thetav_j/\theta_{\cdotv j}$, where $\sum_{v=1}^V \phi_{vk}=1$ and $\theta_{\cdotv j} :=\sum_{k=1}^K \theta_{kj}$, 
a standard 
 procedure is to associate $x_{ji}\in\{1,\ldots,V\}$, the $i$th 
index (word)  of sample $j$, with factor (topic) $z_{ji}\in\{1,\ldots,K\}$, 
and generate 
a bag of indices $\{x_{j1},\ldots,x_{jn_j}\}$ as
\beq
x_{ji}\sim\mbox{Cat}(\phiv_{z_{ji}}),~z_{ji}\sim\mbox{Cat}({\thetav_j}/{\theta_{\cdotv j}}), 
\label{eq:LDA-Po0}
\eeq
 where $n_j=\sum_{v=1}^V n_{vj}$ and $n_{vj}=\sum_{i=1}^{n_j}\delta(x_{ji}=v)$.
We refer to 
 (\ref{eq:LDA-Po0}) as the multinomial mixed-membership model. LDA completes the multinomial mixed-membership model by imposing the Dirichlet priors on both $\{\phiv_k\}_k$ and $\{\thetav_j/
\theta_{\cdotv j}\}_j$ \citep{LDA}.

If in additional to the multinomial mixed-membership model described in (\ref{eq:LDA-Po0}), 
one further generates
 the sample lengths 
 as
\beq
n_{j}\sim\mbox{Pois}(\theta_{\cdotv j}),\label{eq:LDA-Po}
\eeq
then the joint likelihood of $\xv_j: = (x_{j1},\ldots,x_{jn_j})$, $\zv_j: =(z_{j1},\ldots,z_{jn_j})$, and $n_j$ given $\Phimat$ and $\thetav_j$ can be expressed as
$$
P(\xv_{j},\zv_{j}, n_j\,|\,\Phimat,\thetav_j)
= \frac{\prod_{v=1}^V\prod_{k=1}^Kn_{vjk}!}{n_j!} \prod_{v=1}^V\prod_{k=1}^K\mbox{Pois}(n_{vjk};\phi_{vk}\theta_{kj}), 
$$
whose product with the combinatorial coefficient ${n_j!}/( {\prod_{v=1}^V\prod_{k=1}^Kn_{vjk}!})$ becomes the same as the likelihood $P(\{n_{vj},n_{vj1},\ldots,n_{vjK}\}_{v}\,|\,\Phimat,\thetav_j)$ of (\ref{eq:PFA1}).
%

%
%

From the viewpoint of PFA, shown in (\ref{eq:PFA1}), and its 
alternative representation constituted by (\ref{eq:LDA-Po0}) and (\ref{eq:LDA-Po}), 
 a wide variety of discrete latent variable models, 
 such as nonnegative matrix factorization (NMF) \citep{NMF}, PLSA,
 LDA, 
 the gamma-Poisson model of \citet{CannyGaP}, and the discrete component analysis of \citet{DCA}, 
all have the same mechanism to 
model the covariate counts that they
generate both the covariate and factor indices using the categorical distributions shown in (\ref{eq:LDA-Po0}); 
they mainly differ from each other on how the priors on $\phiv_k$ and $\thetav_j$ (or $\thetav_j/\theta_{\cdotv j}$) are constructed \citep{BNBP_PFA_AISTATS2012, NBP2012}. 

\section{Negative binomial factor analysis and Dirichlet-multinomial mixed-membership model}\label{sec:NBFA}

\subsection{Negative binomial factor analysis}

To better model overdispersed counts, instead of following PFA to factorize the covariate-sample count matrix under the Poisson likelihood, we propose negative binomial (NB) factor analysis (NBFA) to factorize it under the NB likelihood as
 \beq
\nv_{j}\sim\mbox{NB}(\Phimat\thetav_{j}, \,p_j), \notag 
\eeq
where $n\sim\mbox{NB}(r,p)$ denote the NB distribution 
with probability mass function (PMF)
$f_N(n\,|\,r,p)=\frac{\Gamma(n+r)}{n!\Gamma(r)} p^n(1-p)^r,$ where $n\in\{0,1,\ldots\}$.
NBFA has an augmented representation as
\beq
n_{vj}=\sum_{k=1}^K n_{vjk}, ~~n_{vjk}\sim\mbox{NB}(\phi_{vk}\theta_{kj},\, p_j).\label{eq:NBFA1}
\eeq
Similar to how the Poisson distribution is related to the multinomial distribution (\emph{e.g.}, \citet{Dunson05bayesianlatent} and Lemma 4.1 of \citet{BNBP_PFA_AISTATS2012}), we reveal how the NB distribution is related to the Dirichlet-multinomial distribution using Theorem \ref{lem:NB_DirMult} shown below, whose proof is provided in Appendix \ref{app:proof}. 
As in \citet{mosimann1962compound}, 
 marginalizing out $\thetav$ from 
$\zv\sim\prod_{i=1}^n\mbox{Cat}(z_i;\thetav),~\thetav\sim\mbox{Dir}(r_1,\ldots,r_K)$
leads to a Dirichlet-categorical (DirCat) distribution with PMF
\beq
P(\zv\,|\,r_1,\ldots,r_K) = \frac{\Gamma(r_{\cdotv})}{\Gamma(n+r_{\cdotv})}\prod_{k=1}^K \frac{\Gamma(n_k+r_k)}{\Gamma(r_k)},\label{eq:like_DirCat}
\eeq
where $n_k=\sum_{i=1}^n \delta(z_i=k)$,
and a Dirichlet-multinomial (DirMult) distribution with PMF $P[(n_1,\ldots,n_K)\,|\,r_1,\ldots,r_K] = \frac{n!}{\prod_{k=1}^K n_k!}P(\zv\,|\,r_1,\ldots,r_K).$


\begin{thm}[The negative binomial and Dirichlet-multinomial distributions]\label{lem:NB_DirMult}
Let $\xv=(x,x_1,$ $\ldots,x_K)$ be random variables generated as
$$x=\sum_{k=1}^K x_k,~x_k\sim\emph{\mbox{NB}}(r_k,p).$$ Set $r_{\cdotv}=\sum_{k=1}^K r_k$ and let $\yv=(y,y_1,\ldots,y_K)$ be random variables generated as
\beq
(y_1,\ldots,y_K)\sim\emph{\mbox{DirMult}}(y,r_1,\ldots,r_K),~y\sim\emph{\mbox{NB}}(r_{\cdotv},p).\notag
\eeq
Then the distribution of $\xv$ is the same as that of $\yv$.
\end{thm}

Using Theorem \ref{lem:NB_DirMult}, 
$(n_{vj}, n_{vj1},\ldots,n_{vjK})$ in (\ref{eq:NBFA1}) can be equivalently generated as
\beq\small
 (n_{vj1},\ldots,n_{vjK})\,|\,n_{vj} \sim\mbox{DirMult}\left(n_{vj},\phi_{v1}\theta_{1j},\ldots,\phi_{vK}\theta_{Kj}\right),~~n_{vj}\sim\mbox{NB}\left( \sum_{k=1}^K \phi_{vk}\theta_{kj}, ~p_j\right).\label{eq:NBFA2}
\eeq 
Clearly, how the factorization under the NB likelihood is related to the Dirichlet-multinomial distribution, as in (\ref{eq:NBFA1}) and (\ref{eq:NBFA2}), mimics how the factorization under the Poisson likelihood is related to the multinomial distribution, as in (\ref{eq:PFA1}) and (\ref{eq:PFA2}).

\subsection{The Dirichlet-multinomial mixed-membership model}

\begin{figure}[!tb]
\begin{center}
\includegraphics[width=0.55\columnwidth]{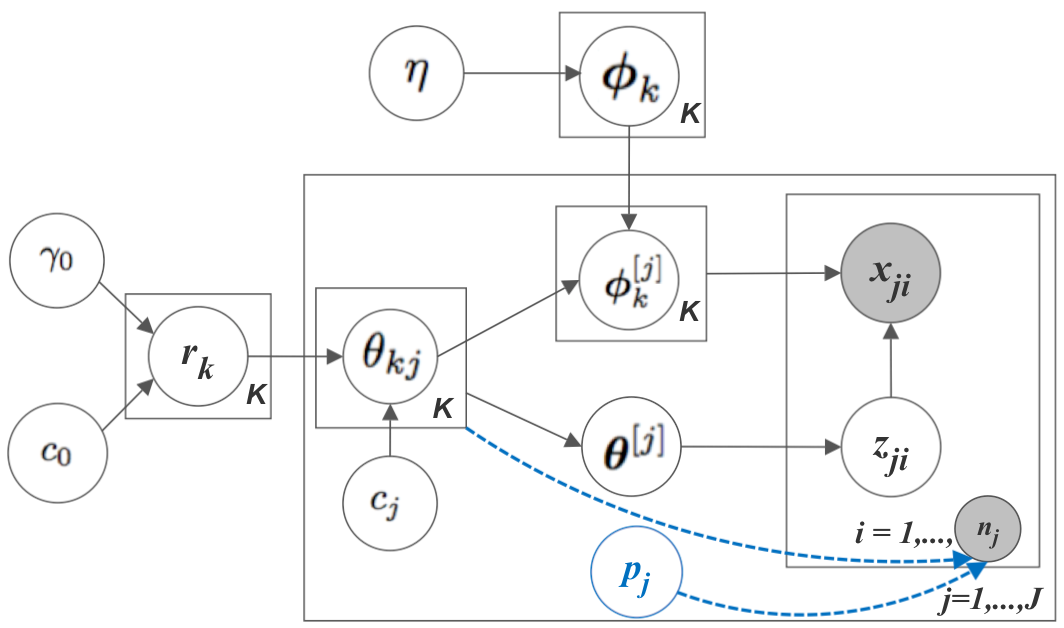}
\end{center}
\caption{\small \label{fig:graph}
Graphical representations of the Dirichlet-multinomial mixed-membership model (without $p_j$ and the dashed lines) and negative binomial factor analysis (with $p_j$ and the dashed lines).
}
\vspace{-3mm}
\end{figure}

Similar to how we relate PFA in (\ref{eq:PFA1}) to the multinomial topic model in (\ref{eq:LDA-Po0}), as in Section \ref{sec:2.2}, 
 we may relate NBFA in (\ref{eq:NBFA1}) to a mixed-membership model constructed by replacing the categorical distributions in \eqref{eq:LDA-Po0} with the Dirichlet-categorical distributions~as
\begin{align}
&x_{ji}\sim\mbox{Cat}(\phiv^{[j]}_{z_{ji}}),~z_{ji}\sim\mbox{Cat}(\thetav^{[j]}),\notag\\
&\phiv^{[j]}_k\sim\mbox{Dir}(\phiv_k \theta_{kj}),~\thetav^{[j]}\sim\mbox{Dir}(\thetav_j),
\label{eq:NBPF3}
\end{align}
where $\{\phiv^{[j]}_k\}_k$ and $\thetav^{[j]}$ represent the factors and factor scores specific for sample~$j$, respectively. A graphical representation of the model, including the settings of the hyper-priors to be discussed later,  is shown in Figure \ref{fig:graph}. Introducing $\phiv^{[j]}_k$ into the hierarchical model allows the same set of factors $\{\phiv_k\}_k$ to be manifested differently in different samples, whereas introducing $\thetav^{[j]}$ allows each sample to have two different representations: the factor scores $\thetav^{[j]}$ under the sample-specific factors $\{\phiv^{[j]}_k\}_k$, and the factor scores $\thetav_j$ under the shared factors $\{\phiv_k\}_k$.
In addition, under this construction, the variance-to-mean ratio of $\phi^{[j]}_{vk}$ given $\phiv_k$ and $\theta_{kj}$ becomes 
$({1-\phi_{vk}})/({\theta_{kj}+1}),$
which monotonically decreases as the corresponding 
factor score $\theta_{kj}$ increases, allowing the variability of $\phiv^{[j]}_{k}$ in the prior to be controlled by the popularity of $\phiv_k$ in the corresponding sample. Moreover, this construction helps simplify 
the model likelihood and allows the model to be closely related to NBFA, as discussed below.

Explicitly drawing $\{\phiv^{[j]}_k\}_k$ for all samples would be computationally prohibitive, especially if the number of samples is large. Fortunately, this operation is totally unnecessary. 
By marginalizing out $\phiv^{[j]}_k$ and $\thetav^{[j]}$ in (\ref{eq:NBPF3}), we have 
\beq
\{x_{ji}\}_{i:z_{ji}=k}\sim\mbox{DirCat}(n_{\cdotv jk},\phi_{1k} \theta_{kj},\ldots,\phi_{Vk} \theta_{kj}),~~\zv_j\sim\mbox{DirCat}(n_j,\thetav_j), \label{eq:NBPF3_0}
\eeq
where $\zv_j:=(z_1,\ldots,z_{jn_j})$, $n_{vjk}:=\sum_{i=1}^{n_j}\delta(x_{ji}=v,z_{ji}=k)$, and $n_{\cdotv jk}:= \sum_{v=1}^V n_{vjk}$. 
 Thus the joint likelihood of $\xv_j:=(x_{j1},\ldots,x_{jn_j})$ and $\zv_j$ given $\Phimat$, $\thetav_j$, and $n_j$ can be expressed as
\beq
P(\xv_j,\zv_j\,|\,\Phimat,\thetav_j,n_j)=\frac{\Gamma(\theta_{\cdotv j})}{\Gamma(n_j+\theta_{\cdotv j})}\prod_{v=1}^V\prod_{k=1}^K \frac{\Gamma(n_{vjk}+\phi_{vk}\theta_{kj})}{\Gamma(\phi_{vk}\theta_{kj})}.\label{eq:DirMult_like}
\eeq
We call the model shown in (\ref{eq:NBPF3}) or (\ref{eq:NBPF3_0}) as the Dirichlet-multinomial mixed-membership (DMMM) model, whose likelihood given the factors and factor scores is shown in (\ref{eq:DirMult_like}). 

We introduce the following proposition, with proof provided in Appendix \ref{app:proof}, to show that one can recover NBFA from the DMMM model by randomizing its sample lengths with the NB distributions, and can reduce NBFA to the DMMM model by conditioning on these lengths. Thus how NBFA and the DMMM are related to each other mimics how PFA and the multinomial mixed-membership model are related to each other.

%
%
%

\begin{prop}[Dirichlet-multinomial mixed-membership (DMMM) modeling and negative binomial factor analysis] \label{lem:DirMult}

For the DMMM model that generates the covariate and factor indices using~(\ref{eq:NBPF3}) or (\ref{eq:NBPF3_0}), 
%
%
if the sample lengths are randomized as
\begin{align}
&n_j\sim\emph{\mbox{NB}}(\theta_{\cdotv j},p_j), \label{eq:NBPF4}
\end{align}
 then the joint likelihood of $\xv_j$, $\zv_j$, and $n_j$ given $\Phimat$, $\thetav_j$, and $p_j$, %
%
multiplied by the combinatorial coefficient ${n_j!}/({\prod_{v=1}^V\prod_{k=1}^K n_{vjk}!})$, is equal to the likelihood 
of negative binomial factor analysis (NBFA) described in (\ref{eq:NBFA1}) or (\ref{eq:NBFA2}), 
expressed as
\beq
P(\{n_{vj},n_{vj1},\ldots,n_{vjK}\}_{v=1,V}\,|\,\Phimat,\thetav_j,p_j)=\prod_{v=1}^V\prod_{k=1}^K \emph{\mbox{NB}}(n_{vjk}; \phi_{vk}\theta_{kj}, p_j). \label{eq:NBPF5}
\eeq
\end{prop}

The DMMM model could model not only the burstiness of the covariates, but also that of the factors via the Dirichlet-categorical distributions, as further explained below when discussing related models. As far as the conditional posteriors of $\phiv_k$ and $\thetav_j$ are concerned, the DMMM model with the lengths of its samples randomized via the NB distributions, as shown in (\ref{eq:NBPF3_0}) and (\ref{eq:NBPF4}), is equivalent to NBFA, as shown in~(\ref{eq:NBFA1}). 
The representational differences, 
however, lead to different inference strategies, which will be discussed in detail along with their 
nonparametric Bayesian generalizations. 

\subsection{Comparisons with related models}\label{sec:3.3} 

Preceding the Dirichlet-multinomial mixed-membership (DMMM) model proposed in this paper, to account for covariate burstiness, \citet{doyle2009accounting} proposed Dirichlet compound multinomial LDA (DCMLDA) that can be expressed as
\begin{align}
&x_{ji}\sim\mbox{Cat}(\phiv^{[j]}_{z_{ji}}),~~z_{ji}\sim\mbox{Cat}(\thetav_j),~~
\phiv^{[j]}_k\sim\mbox{Dir}(\phiv_k r_k), 
\label{eq:NBPF3_4}
\end{align}
where the Dirichlet prior is further imposed on $\thetav_j$. Note that the sample-specific factor scores $\{\thetav_j\}_j$ are represented under the sample-specific factors $\{\phiv^{[j]}_k\}_k$ in DCMLDA, as shown in (\ref{eq:NBPF3_4}), whereas they are represented under the same set of factors $\{\phiv_k\}_k$ in the DMMM model, as shown in (\ref{eq:NBPF3_0}).

Comparing (\ref{eq:NBPF3}) with (\ref{eq:NBPF3_4}), it is clear that removing $\thetav^{[j]}$ from (\ref{eq:NBPF3}) 
reduces the DMMM model 
to DCMLDA in (\ref{eq:NBPF3_4}).
Moreover, if we further let
\begin{align}
&\thetav_j\sim\mbox{Dir}(\rv),~~n_j\sim{\mbox{NB}}(r_{\cdotv},p_j), \label{eq:NBPF3_5}
\end{align}
 then the joint likelihood of $\xv_j$, $\zv_j$, and $n_j$ given $\Phimat$, $\rv$, and $p_j$ can be expressed as
\beq
P(\xv_j,\zv_j,n_j\,|\,\Phimat,\rv,p_j)= 
 \frac{1}{n_j!}\prod_{v=1}^V\prod_{k=1}^K\frac{\Gamma(n_{vjk}+\phi_{vk}r_k)}{\Gamma(\phi_{vk}r_k)}p_j^{n_{vjk}}(1-p_j)^{\phi_{vk}r_{k}}, 
 \label{eq:NBFA_like1}
\eeq
which multiplied by the combinatorial coefficient ${n_j!}/({\prod_{v=1}^V\prod_{k=1}^K n_{vjk}!})$ is equal to the likelihood of $\{n_{vj},n_{vj1},\ldots,n_{vjK}\}_{v=1,V}$ given $\Phimat$,~$\rv$, and $p_j$ in 
\beq
n_{vj}=\sum_{k=1}^K n_{vjk},~n_{vjk}\sim\mbox{NB}(\phi_{vk}r_k, p_j). \label{eq:NBPF3_6}
\eeq
Thus, as far as the conditional posteriors of $\{\phiv_k\}_k$ and $\{r_k\}_k$ are concerned, DCMLDA constituted by (\ref{eq:NBPF3_4})-(\ref{eq:NBPF3_5}) is equivalent to a special case of NBFA shown in (\ref{eq:NBPF3_6}), which is the augmented representation of $\nv_j\sim\mbox{NB}(\Phimat\rv,p_j)$ that restricts all samples to have the same factor scores $\{r_k\}_k$ under the same set of shared factors $\{\phiv_k\}_k$.

Given the factors $\{\phiv_k\}_k$ and factor scores $\thetav_j$, for the multinomial mixed-membership model in (\ref{eq:LDA-Po0}), both the covariate and factor indices are independently drawn from the categorical distributions; for DCMLDA in (\ref{eq:NBPF3_4}), the factor indices but not the covariate indices are independently drawn from the categorical distributions; whereas for the DMMM model in (\ref{eq:NBPF3}), 
neither the factor indices nor covariate indices are independently drawn from the categorical distributions. 
Denoting $y^{-ji}$ as the variable $y$ obtained by excluding the contribution of the $i$th word 
in sample~$j$, 
we compare in Table~\ref{tab:1} these three different models on their predictive distributions of $x_{ji}$ and $z_{ji}$. 
\begin{table}
\begin{small}
\begin{center}
\caption{Comparisons of the predictive distributions 
for the multinomial mixed-membership model in (\ref{eq:LDA-Po0}), Dirichlet compound multinomial LDA (DCMLDA) in (\ref{eq:NBPF3_4}), and Dirichlet-multinomial mixed-membership (DMMM) model in (\ref{eq:NBPF3}). }\label{tab:1}
\small
\begin{tabular}{ C{8pc}C{10pc} C{10pc} }
\toprule
Model & Predictive distribution for $x_{ji}$ & Predictive distribution for $z_{ji}$ \\
\toprule
Multinomial mixed-membership & $P(x_{ji}=v\,|\,\Phimat,\thetav_j) = \phi_{vz_{ji}}$ & $P(z_{ji}=k\,|\,\thetav_j) = {\theta_{kj}}/{\theta_{\cdotv j}}$ \\ 
\midrule
Dirichlet compound multinomial
LDA & $\displaystyle P(x_{ji}=v\,|\,\xv_j^{-ji},z_{ji},\Phimat,\rv) = \frac{n_{vjz_{ji}}^{-ji}+\phi_{vz_{ji}}r_{z_{ji}}}{n_{\cdotv jz_{ji}}^{-ji} +r_{z_{ji}}}$ & $P(z_{ji}=k\,|\,\thetav_j) = {\theta_{kj}}/{\theta_{\cdotv j}}$ \\
\midrule
Dirichlet-multinomial mixed-membership &$\displaystyle P(x_{ji}=v\,|\,\xv_j^{-ji},z_{ji},\Phimat,\thetav_j) = \frac{n_{vjz_{ji}}^{-ji}+\phi_{vz_{ji}}\theta_{z_{ji}j}}{n_{\cdotv jz_{ji}}^{-ji} +\theta_{z_{ji}j}}$ & $\displaystyle P(z_{ji}=k\,|\,\zv_j^{-i},n_j,\thetav_j) = \frac{n_{\cdotv jk}^{-ji}+\theta_{kj}}{n_{j}-1+\theta_{\cdotv j}}$\\
 \bottomrule
\end{tabular}
\end{center}\vspace{-7.5mm}
\end{small}
\end{table}%

In comparison to the multinomial mixed-membership model, DCMLDA allows the number of times that a covariate 
 appears in a sample to exhibit the ``rich get richer'' ($i.e.$, self-excitation) behavior, leading to a better modeling of covariate burstiness, and the DMMM model further models the burstiness of the factor indices and hence encourages not only self-excitation, but also cross-excitation of covariate frequencies. 
 It is clear from Table~\ref{tab:1} that DCMLDA models covariate burstiness but not factor burstiness, and the corresponding NBFA restricts all samples to have the same factor scores under the shared factors $\{\phiv_k\}_k$. Thus we expect the DMMM model to clearly outperform DCMLDA, as will be confirmed by our experimental results.

Note that  \citet{GBN} have recently extended the single-layer PFA into a multilayer one. For example,  a two-layer PFA would further factorize the factor scores $\thetav_j$ in $\xv_j\sim\mbox{Pois}(\Phimat \thetav_j)$ under the gamma likelihood as
$\thetav_j\sim\mbox{Gamma}[\tilde{\Phimat} \tilde{\thetav}_j,p_j/(1-p_j)]$, which is designed to capture the co-occurrence patterns of the first-layer factors $\phiv_k$. With $\theta_{jk}$ marginalized out from $n_{\cdotv jk}\sim\mbox{Pois}(\theta_{jk})$, we have $n_{\cdotv jk}\sim\mbox{NB}(\sum_{\tilde k}\tilde{\phi}_{k\tilde k} \tilde{\theta}_{j\tilde k}, p_j)$, which can be considered as a NBFA model for the latent factor-sample counts. Thus using our previous analysis on NBFA, this multilayer construction models both the self- and cross-excitations of the first-layer factors and helps capture their bustiness, which could help better explain why adding an additional layer to PFA could often lead to clearly improved modeling performance, as shown in \citet{GBN}. However, due to the choice of the Poisson likelihood at the first layer, a multilayer PFA does not directly capture the bustiness of the covariates. On the other hand, the single-layer NBFA models the self-excitation but not cross-excitation of factors. Therefore, extending the single-layer NBFA into a multi-layer one by exploiting the deep structure of \citet{GBN}  may combine the advantages of both  by not only capturing the covariate bustiness, but also better capturing the factor bustiness. That extension is outside of the scope of the  paper and we leave it for future study. 

%

\section{Hierarchical gamma-negative binomial process negative binomial factor analysis}\label{sec:hGNBP}


Let 
$G\sim\Gamma{\mbox{P}}(G_0,1/c_0)$ be a gamma process \citep{ferguson73} defined on the product space $\mathbb{R}_+\times \Omega$, where $\mathbb{R}_+=\{x:x>0\}$, with two parameters: a finite and continuous base measure $G_0$ over a complete separable metric space $\Omega$, and a scale $1/c_0$, such that $G(A)\sim\mbox{Gamma}(G_0(A),1/c_0)$ for each $A\subset \Omega$. The L\'evy measure of the gamma process can be expressed as $\nu(drd\phiv)=r^{-1}e^{-c_0r}dr G_0(d\phiv)$. 
A draw from 
$G\sim\Gamma\mbox{P}(G_0,1/c_0)$ can be represented as the countably infinite sum
$G=\sum_{k=1}^\infty r_k\delta_{\phiv_k},~\phiv_k\sim g_0,$
where $\gamma_0=G_0(\Omega)$ as the mass parameter and $g_0(d\phiv)=G_0(d\phiv)/\gamma_0$ is the base distribution.

%

To support countably infinite factors for the DMMM model, we consider 
a hierarchical gamma-negative binomial process (hGNBP) as
\beq
X_j\sim\mbox{NBP}(\Theta_j,p_j),~\Theta_j\sim\Gamma\mbox{P}(G,1/c_j),~G\sim\Gamma\mbox{P}(G_0,1/c_0), \notag
\eeq
%
%
where 
$X\sim\mbox{NBP}(\Theta,p)$ is a  NB process defined such that $X(A_i)\sim\mbox{NB}(\Theta(A_i),p)$ are independent NB random variables for disjoint partitions $A_i$ of $\Omega$. Given a gamma process random draw $G = \sum_{k=1}^\infty r_k\delta_{\phiv_k}$, we have
\beq
\Theta_j=\sum_{k=1}^\infty \theta_{kj}\delta_{\phiv_k},~\theta_{kj}\sim\mbox{Gamma}(r_k,1/c_j),\notag
\eeq
where $\theta_{kj}:=\Theta_j(\phiv_k)$ measures the weight of factor $k$ in sample $j$,
and
\beq
X_j = \sum_{k=1}^\infty n_{jk}\delta_{\phiv_k},~n_{jk}\sim\mbox{NB}(\theta_{kj},p_j),\notag
\eeq
where $n_{j} =X_j(\Omega)$ is the length of sample $j$ and $n_{jk}:=X_j(\phiv_k)=\sum_{i=1}^{n_j}\delta(z_{ji}=k)$ represents the number of times that factor $k$ appears in sample $j$.

We provide posterior analysis for the proposed hGNBP in Appendix \ref{app:post}, where we utilize several additional  discrete distributions, including the Chinese restaurant table (CRT), logarithmic, and sum-logarithmic (SumLog) distributions, that will also be used  in the following discussion. 

\subsection{Hierarchical model} 

We express the hGNBP-DMMM model as
\begin{align}
&x_{ji}\sim\mbox{Cat}(\phiv^{[j]}_{z_{ji}}),~z_{ji}\sim\mbox{Cat}(\thetav^{[j]}),\notag\\
&\phiv^{[j]}_k\sim\mbox{Dir}(\phiv_k \theta_{kj}),~\thetav^{[j]}\sim\mbox{Dir}(\thetav_j),\notag\\
&\theta_{kj}\sim\mbox{Gamma}(r_k,1/c_j),~c_j\sim\mbox{Gamma}(e_0,1/f_0),\notag\\
&n_j\sim\mbox{NB}(\theta_{\cdotv j},p_j), ~p_j\sim\mbox{Beta}(a_0,b_0),~G\sim\Gamma\mbox{P}(G_0,1/c_0), \label{eq:full model}
\end{align}
where the atoms of the gamma process are drawn from a Dirichlet base distribution
$
\phiv_k\sim\mbox{Dir}(\eta,\ldots,\eta).
$
We further let $\gamma_0\sim\mbox{Gamma}(a_0,1/b_0)$ and $c_0\sim\mbox{Gamma}(e_0,1/f_0)$. 
With Proposition \ref{lem:DirMult}, the above 
model 
can also be represented as a hGNBP-NBFA as
\begin{align}
&n_{vj}=\sum_{k=1}^\infty n_{vjk}, ~~n_{vjk}\sim\mbox{NB}(\phi_{vk}\theta_{kj},\, p_j),\notag\\
&\theta_{kj}\sim\mbox{Gamma}(r_k,1/c_j),~c_j\sim\mbox{Gamma}(e_0,1/f_0),\notag\\
&p_j\sim\mbox{Beta}(a_0,b_0),~G\sim\Gamma\mbox{P}(G_0,1/c_0).\label{eq:NBFA1_2}
\end{align}
The DMMM model in (\ref{eq:full model}) and NBFA in (\ref{eq:NBFA1_2}) have the same conditional posteriors for both the factors $\{\phiv_k\}_k$ and factor scores $\{\thetav_j\}_j$, but lead to different inference strategies. To infer $\{\phiv_k\}_k$ and $\{\thetav_j\}_j$, we first develop both blocked and collapsed Gibbs sampling with (\ref{eq:full model}), 
and then develop blocked Gibbs sampling with~\eqref{eq:NBFA1_2}, as described below. 

\section{Inference via Gibbs sampling}\label{sec:inference}
We present in this section three different Gibbs sampling algorithms, as summarized in Algorithm \ref{tab:algorithm} in the Appendix. 


\subsection{Blocked Gibbs sampling} \label{sec:blockGibbs}
As it is impossible to draw all the countably infinite atoms of a gamma process draw, expressed as $G=\sum_{k=1}^\infty r_k\delta_{\phiv_k}$, for the convenience of implementation, it is common to consider truncating the total number of atoms to be $K$ by choosing a discrete base measure as $G_0=\sum_{k=1}^K \frac{\gamma_0}K\delta_{\phiv_k}$, under which we have
$r_k\sim\mbox{Gamma}(\gamma_0/K,1/c_0)$ for $k\in\{1,\ldots,K\}$ \citep{NBP2012}. The finite truncation strategy is also commonly used for 
Dirichlet process mixture models \citep{
ishwaran2001gibbs,fox2011sticky}. 
Although fixing 
$K$ often works well in practice, it may lead to a considerable waste of computation if $K$ is set to be too large. 
For nonparametric Bayesian mixture models based on the Dirichlet process \citep{ferguson73,Escobar1995} or other more general normalized random measures with independent increments (NRMIs) \citep{regazzini2003distributional,lijoi2007controlling}, 
one may consider slice sampling to adaptively truncate the number of atoms used in each Markov chain Monte Carlo (MCMC) iteration \citep{walker2007sampling,papaspiliopoulos2008retrospective}. 
Unlike NRMIs whose atoms' weights have to sum to one and hence are negatively correlated, the weights of the atoms of completely random measures are independent from each other.
Exploiting this property, for our models built on completely random measures, we 
 construct a sampling procedure that adaptively truncates the total number of atoms in each iteration. 

Note that the conditional posterior of $G$ 
shown in (\ref{eq:post_analysis}) consists of 
 two independent gamma processes: $\mathcal{D}\sim\Gamma\mbox{P}\big(\sum_j \widetilde{L}_j,[{c_0-\sum_j\ln(1-\tilde{p}_j)}]^{-1}\big)$ and $G_\star\sim\Gamma\mbox{P}\big(G_0,[c_0-\sum_j\ln(1-\tilde{p}_j)]^{-1}\big)$. 
To approximately represent a draw from $G_\star$ 
that consists of countably infinite atoms, at the end of each MCMC iteration,
we relabel the indices of the atoms with nonzero counts from 1 to $K^+:=\sum_{k=1}^\infty \delta(\sum_j X_{j}(\phiv_k)>0)$; draw $K_\star$ new atoms 
 as 
$$\widetilde{G}_\star = \sum_{k=K^++1}^{K^++K_\star} r_{k}\delta_{\phiv_k},~r_k\sim\mbox{Gamma}\left(\frac{\gamma_0}{K_\star}, \frac{1}{c_0-\sum_j\ln(1-\tilde{p}_j)}\right),~\phiv_k\sim\mbox{Dir}(\eta,\ldots,\eta);$$
and set $K:=K^++K_\star$ as the total number of atoms to be used for the next iteration. 

For the hGNBP DMMM model, we present the update equations for $z_{ji}$ and $\ell_{vjk}$ below and those for all the other parameters in Appendix \ref{app:hGNBP_1}.\\
\emph{\textbf{Sample $z_{ji}$.}}
Using the likelihood in (\ref{eq:DirMult_like}), 
we have 
\beq
P(z_{ji}=k\,|\,x_{ji},\zv_{j}^{-i},\Phimat, \thetav_j)\propto n_{x_{ji}jk}^{-ji}+\phi_{x_{ji}k}\theta_{kj},~~~k\in\{1,\ldots,K\}. \label{eq:samplezji}
\eeq
\emph{\textbf{Sample $\ell_{vjk}$.}}
Since $n_{vjk}\sim\mbox{NB}(\phi_{vk}\theta_{kj},p_j)$ in the prior, as shown in Proposition \ref{lem:DirMult}, we can draw a corresponding latent count $\ell_{vjk}$ for each $n_{vjk}$ as
\beq
(\ell_{vjk}\,|\,-)\sim \mbox{CRT}(n_{vjk}, \phi_{vk}\theta_{kj}),\label{eq:ellvjk}
\eeq
where $\ell_{vjk}=0$ almost surely if $n_{vjk}=0$.



\subsection{Collapsed Gibbs sampling}\label{sec:collapseGibbs}
 
 Let us denote $\bv\sim\mbox{CRP}(n,r)$ as an exchangeable random partition of the set $\{1,\ldots,n\}$, generated by assigning $n$ customers (samples) to $\ell$ random number of tables (exclusive and nonempty subsets) using a Chinese restaurant process (CRP) \citep{aldous:crp} with concentration parameter~$r$. The exchangeable partition probability function of $\bv$ under the CRP, also known as the Ewens sampling formula \citep{ewens1972sampling,DP_Mixture_Antoniak}, can be expressed as
 $P(\bv\,|\, n,r) = \frac{\Gamma(r)r^\ell}{\Gamma(n+r)}\prod_{t=1}^\ell \Gamma(n_t),$
 where $n_t=\sum_{i=1}^n\delta(b_i=t)$ is the size of the $t$th subset and $\ell$ is the number of subsets \citep{csp}.
 
One common strategy to improve convergence and mixing for multinomial mixed-membership models is to collapse the factors $\{\phiv_k\}_k$ and factor scores $\{\thetav_{j}\}_j$ in the sampler \citep{FindSciTopic,newman2009distributed}. To apply this strategy to the DMMM model, we first need to transform the likelihood in (\ref{eq:NBFA_like1}) to make it amenable to marginalization. Using an analogy similar to that for the Chinese restaurant franchise of \citet{HDP}, if we consider $z_{ji}$ as the index of the ``dish'' that the $i$th ``customer'' in the $j$th ``restaurant'' takes, then, to make the likelihood in (\ref{eq:NBFA_like1}) become fully factorized, we may introduce $b_{ji}$ as the index of the table at which this customer is seated. The following proposition, whose proof is provided in Appendix A, reveals how the CRP can be related to the Dirichlet-multinomial and NB distributions, and shows how to introduce auxiliary variables to make the likelihood of $\zv\sim\mbox{DirCat}(n,r_1,\ldots,r_K)$, as shown in~(\ref{eq:like_DirCat}), become fully factorized.

\begin{prop} 
\label{lem:CRP_DirCat}
Given the sample length $n$ (number of customers) and $\rv=(r_1,\ldots,r_K)$, 
the joint distribution of the ``table'' indices $\bv=(b_1,\ldots,b_n)$ and ``dish'' indices $\zv=(z_1,\ldots,z_n)$ 
in
\beq
\{b_i\}_{i:z_i=k}\sim\emph{\mbox{CRP}}(n_{k},r_k),~\zv\sim\emph{\mbox{DirCat}}(n,r_1,\ldots,r_K), \notag
\eeq
is the same as that in 
\beq
z_i = s_{b_i},~ s_t\sim\emph{\mbox{Cat}}(r_1/r_{\cdotv},\ldots,r_K/r_{\cdotv}),~\bv\sim\emph{\mbox{CRP}}(n,r_{\cdotv}),\notag 
\eeq
with PMF
\begin{align}
P(\bv,\zv\,|\,n,\rv) =\frac{\Gamma(r_{\cdotv})}{\Gamma(n+r_{\cdotv})}\prod_{k=1}^{K} \left\{ r_k^{\ell_k}\prod_{t=1}^{\ell_k}\Gamma(n_{kt}) \right\},\notag
\end{align}
where $\ell_k $ is the number of unique indices in $\{b_i\}_{i:z_i=k}$, $\ell_{\cdotv}=\sum_{k=1}^K \ell_k$ is the total number of nonempty tables, $t=1,\ldots,\ell_{\cdotv}$, and $n_{kt} = \sum_{i=1}^n \delta(b_i=t,z_i=k)$ is the number of customers that sit at table $t$ and take dish $k$.

If we further randomize the sample length as
$$n\sim\emph{\mbox{NB}}(r_{\cdotv},p),$$
then we have 
the PMF for the joint distribution of $\bv$, $\zv$, and $n$ given $\rv$ and $p$ in a fully factorized form as
\begin{align}
P(\bv,\zv,n\,|\,\rv,p) &= \frac{1}{n!}\prod_{k=1}^{K} \left\{ r_k^{\ell_k}(1-p)^{r_k}p^{n_k}\prod_{t=1}^{\ell_k}\Gamma(n_{kt}) \right\},\notag
\end{align}
which, with appropriate combinatorial analysis, can be mapped to
the PMF of the joint distribution of $\nv=(n_1,\ldots,n_k)$, $\lv=(l_1,\ldots,\ell_k)$, and $n$ given $\rv$ and $p$ in
\beq
n=\sum_{k=1}^Kn_k,~n_k=\sum_{t=1}^{\ell_k} n_{kt}, ~n_{kt}\sim\emph{\mbox{Logarithmic}}(p), ~\ell_k\sim\emph{\mbox{Pois}}[-r_k\ln(1-p)]. \label{eq:PoissonLogBig}
\eeq
\end{prop}

Using (\ref{eq:NBFA_like1}) and Proposition \ref{lem:CRP_DirCat}, introducing the auxiliary variables
\beq
\{b_{ji}\}_{i:\,x_{ji}=v,z_{ji}=k}\sim\mbox{CRP}(n_{vjk},\phi_{vk}\theta_{kj}), \label{eq:aug_bji}
\eeq
 we have
the joint likelihood 
 for the DMMM model~as
\beq
P(\bv_j,\xv_j,\zv_j,n_j\,|\,\Phimat,\thetav_j,p_j)=\frac{1}{n_j!}\prod_v\prod_{k} \left\{{(\phi_{vk}\theta_{kj})^{\ell_{vjk}}}{ p_j^{n_{vjk}}(1-p_j)^{\phi_{vk}\theta_{kj}}} \prod_{t=1}^{\ell_{vjk}}\Gamma(n_{vjkt})\right\}, \label{eq:fullyaugmented}
\eeq
where $\ell_{vjk} $ is the number of unique indices in $\{b_{ji}\}_{i:\,x_{ji}=v,z_{ji}=k}$ and $n_{vjkt} =\sum_{i=1}^{n_j} \delta(x_{ji}=v,z_{ji}=k,b_{ji}=t)$.
As in Proposition \ref{lem:CRP_DirCat}, instead of first assigning the indices to factors using the Dirichlet-multinomial distributions and then assigning the indices with the same factors to tables, we may first assign the indices  to tables and then assign the tables to factors. Thus we have the following model 
\begin{align}
&x_{ji} = w_{jz_{ji}},~
z_{ji} = s_{jb_{ji}}, 
~w_{js_{jt}}\sim\mbox{Cat}(\phiv_{s_{jt}}), ~
s_{jt}\sim\mbox{Cat} ({\thetav_{j}}/{\theta_{\cdotv j}}),\notag\\
&~~~~~~~~~~~~\bv_j\sim\mbox{CRP}(n_j,\theta_{\cdotv j}),~n_j\sim\mbox{NB}(\theta_{\cdotv j},p_j),\notag
\end{align}
whose likelihood $P(\bv_j,\xv_j,\zv_j,n_j\,|\,\Phimat,\thetav_j,p_j)$ 
is the same as the likelihood, as shown in (\ref{eq:fullyaugmented}), of the DMMM model constituted of (\ref{eq:NBPF3_0}) and (\ref{eq:NBPF4}) and augmented with (\ref{eq:aug_bji}).

We outline the collapsed Gibbs sampler for the hGNBP-NBFA in Algorithm \ref{tab:algorithm} and 
provide the derivation and update equations below.
This collapsed sampling strategy marginalizes out both the factors $\{\phiv_k\}$ and factor scores $\{\thetav_{j}\}_j$, but at the expense of introducing an auxiliary variable $b_{ji}$ for each index $x_{ji}$. 


%
Marginalizing out $\Phimat$ and $\Thetamat$ from $\prod_j P(\bv_j,\xv_j,\zv_j,n_j\,|\,\Phimat,\thetav_j,p_j)$, we have\small
\begin{align}
&P(\{\bv_j,\xv_j,\zv_j,n_j\}_j\,|\,G,\pv) = ~
e^{ r_{\star} \sum_j \ln(1-\tilde p_j)}\left\{\prod_j p_j^{n_j}\frac{\prod_v\prod_k\left(\prod_{t=1}^{\ell_{vjk}}\Gamma(n_{vjkt})\right)}{n_j!}\right\}\notag\\
&\times\left\{\prod_{k:\,\ell_{\cdotv \cdotv k}>0} \frac{\Gamma(V\eta)}{\Gamma(\ell_{\cdotv\cdotv k}+V\eta)}\prod_{v=1}^V\frac{\Gamma(\ell_{v\cdotv k}+\eta)}{\Gamma(\eta)}\right\} 
\left\{
\prod_j\prod_{k:\,\ell_{\cdotv \cdotv k}>0} \frac{\Gamma(r_k+\ell_{\cdotv jk})}{\Gamma(r_k)} \frac{c_j^{r_k}}{[c_j-\ln(1-p_j)]^{r_k+\ell_{\cdotv jk}}}\right\},\label{eq:Like}
\end{align}\normalsize
where $\ell_{\cdotv \cdotv k} := \sum_j \ell_{\cdotv jk}$ and $r_\star :=\sum_{k:\,\ell_{\cdotv \cdotv k}=0} r_k$.\\
\emph{\textbf{Sample $z_{ji}$ and $b_{ji}$.}}
Using the likelihood in (\ref{eq:Like}), with $(K^+)^{-ji}$ representing the number of active atoms without considering $z_{ji}$, we have \small
\begin{align} 
&P(z_{ji}\!=\!k,b_{ji}\!=\!t\,|\,x_{ji},\zv^{-ji},\bv^{-ji},G)\! 
\propto\!
\begin{cases} \footnotesize
\vspace{.15cm}
\displaystyle n_{x_{ji}jkt}^{-ji}
, \!\!\!&\! \emph{\mbox{if }} k\le (K^+)^{-ji}, t\le \ell_{x_{ji}jk}^{-ji} ; \\ 
 \footnotesize \displaystyle \frac{\ell_{x_{ji}\cdotv k}^{-ji}+\eta}{\ell_{\cdotv\cdotv k}^{-ji}+V\eta}\frac{r_k+\ell_{\cdotv jk}^{-ji}}{c_j-\ln(1-p_j)}, \!\!\!& \!\emph{\mbox{if }} k\le (K^+)^{-ji}, t= \ell_{x_{ji}jk}^{-ji}+1 ;\\
\vspace{.15cm}
 \footnotesize\displaystyle \frac{1}{V}\frac{r_{\star}}{c_j-\ln(1-p_j)},\!\!\!&\! \emph{\mbox{if }} k= (K^+)^{-ji}+1, t=1 ; \end{cases} \notag
\end{align} \normalsize
and if $k= (K^+)^{-ji}+1$ happens, similar to the direct assignment sampler for the hierarchical Dirichlet process \citep{HDP}, we draw $\beta\sim\mbox{Beta}(1,\gamma_0)$ and then let $r_k =\beta r_{\star} $ and $r_{\star} =(1-\beta) r_{\star}$. This is based on the stick-breaking representation of the Dirichlet process, $\widetilde{G}_\star\sim\mbox{DP}\left(\gamma_0,G_0/\gamma_0\right)$, whose product with an independent random variable $r_{\star}\sim\left(\gamma_0,[{c_0-\sum_j\ln(1-\tilde{p}_j)}]^{-1}\right)$ recovers the gamma process $G_\star\sim\Gamma\mbox{P}\left(G_0,[{c_0-\sum_j\ln(1-\tilde{p}_j)}]^{-1}\right)$.
Note that
instead of drawing both $z_{ji}$ and $b_{ji}$ at the same time, one may first draw $z_{ji}$ and then draw $b_{ji}$ given $z_{ji}$, or first draw $b_{ji}$ and then draw $z_{ji}$ given $b_{ji}$. 
We describe how to sample the other parameters in Appendix~\ref{app:hGNBP_2}.

\subsection{Blocked Gibbs sampling under compound Poisson representation 
}\label{sec:blockGibbs_1} 

Examining both the blocked and collapsed Gibbs samplers presented Sections \ref{sec:blockGibbs} and \ref{sec:collapseGibbs}, respectively, and their sampling steps shown in Algorithm \ref{tab:algorithm}, 
one may notice that to obtain $n_{vjk}$ in each iteration, one has to go through all individual indices $x_{ji}$, for each of which the sampling of $z_{ji}$ from a multinomial distribution takes $O(K)$ computation. However, as it is the $\ell_{vjk}$ but not $n_{vjk}$ that are required for sampling 
all the other model parameters, one may naturally wonder whether the step of sampling $z_{ji}$ to obtain $n_{vjk}$ can be skipped. To answer that question, 
we first introduce the following theorem, whose proof is provided in Appendix \ref{app:proof}.

\begin{thm}\label{cor:CRTMult}
Conditioning on $n$ and $\rv$, with $\{n_k\}_k$ marginalized out, the distribution of $\lv=(\ell_1,\ldots,\ell_k)$ 
 in
\beq
\ell_k\,|\,n_k\sim\emph{\mbox{CRT}}(n_{k},r_k),~\nv\,|\, n\sim\emph{\mbox{DirMult}}
(n, r_1,\ldots,r_K) \notag
\eeq
is the same as that in 
\beq
\ellv\,|\,\ell_{\cdotv}\sim\emph{\mbox{Mult}}(\ell_{\cdotv},r_1/r_{\cdotv},\ldots,r_K/r_{\cdotv}),~\ell_{\cdotv}\,|\, n\sim\emph{\mbox{CRT}}(n,r_{\cdotv}), \notag
\eeq
with PMF
$$
P(\ellv\,|\,n,r_1,\ldots,r_K)=\frac{\ell_{\cdotv}!}{\prod_{k=1}^K \ell_k!}\frac{\Gamma(r_{\cdotv})}{\Gamma(n+r_{\cdotv})}|s(n,\ell_{\cdotv})| \prod_{k=1}^{K} r_k^{\ell_k} .
$$
\end{thm}

%
%

Rather than representing NBFA in (\ref{eq:NBFA1_2}) as the DMMM model in (\ref{eq:full model}), we may directly consider its compound Poisson representation as 
\beq
n_{vj}=\sum_{t=1}^{\ell_{vj}} n_{vjt},~ n_{vjt}\sim\mbox{Logarithmic}(p_j), ~\ell_{vj} 
\sim\mbox{Pois}\left[-\sum_k\phi_{vk}\theta_{kj} \ln(1-p_j)\right]. \label{eq:NBFA0_2}
\eeq
Under this representation, 
we may first infer $\ell_{vj}$ for each $n_{vj}$ and then factorize the latent count matrix $\{\ell_{vj}\}_{v,j}$ under the Poisson likelihood, as described below. 
%
%
%
%
%
%

Rather than first sampling $z_{ji}$ (and hence $n_{vjk}$) using (\ref{eq:samplezji}) and then sampling $\ell_{vjk}$ using (\ref{eq:ellvjk}), 
with Theorem 
\ref{cor:CRTMult} 
and the compound Poisson representation in (\ref{eq:NBFA0_2}), we can skip sampling $z_{ji}$ and 
directly sample $\ell_{vjk}$ as
\begin{align}
(\ell_{vj}\,|\,-)&\sim \mbox{CRT}\left(n_{vj}, \sum_k \phi_{vk}\theta_{kj}\right),\label{eq:ellvjk_CRT1}\\
[(\ell_{vj1},\ldots,\ell_{vjK})\,|\,-]&\sim \mbox{Mult}\left(\ell_{vj}, \frac{\phi_{v1}\theta_{1j}}{\sum_k \phi_{vk}\theta_{kj}},\ldots,\frac{\phi_{vK}\theta_{kj}}{\sum_k \phi_{vk}\theta_{kj}}\right).\label{eq:ellvjk_CRT2}
\end{align}
All the other model parameters can be sampled in the same way as they are sampled in 
the regular blocked Gibbs sampler, with
 (\ref{eq:sampletpj})-(\ref{eq:sampletheta}). 

Note that $\ell_{vj}=n_{vj}$ a.s. if $n_{vj}\in\{0,1\}$, $\ell_{vj}\le n_{vj}$ a.s. if $n_{j}\ge 2$, and 
$$\textstyle\E[\ell_{vj}\,|\,n_{vj}, \phi_{vk}\theta_{kj}] =\left( \sum_k \phi_{vk}\theta_{kj}\right) \left[\psi(n_{vj}+\phi_{vk}\theta_{kj} ) - \psi(\phi_{vk}\theta_{kj} )\right],$$ where $\psi(x)=\frac{d}{dx}\ln\Gamma(x)$ is the digamma function. Thus $\E[\ell_{vj}] \approx ( \sum_k \phi_{vk}\theta_{kj}) \ln(n_{vj}) $ when $n_{vj}$ is large. 
Clearly, this new sampling strategy significantly impacts $n_{vj}$ that are large. 
In comparison to the regular blocked Gibbs sampler 
described in detail in Section \ref{sec:blockGibbs}, 
the compound Poisson based blocked Gibbs sampler essentially replaces (\ref{eq:samplezji})-(\ref{eq:ellvjk}) of the regular one 
 with (\ref{eq:ellvjk_CRT1})-(\ref{eq:ellvjk_CRT2}), which can be readily justified by Theorem~\ref{cor:CRTMult}.
 Instead of directly assigning the $n_{vj}$ covariate indices to the $K$ latent factors, now we only need to assign them to $ O[\ln(n_{vj}+1)]$ tables and directly assign these tables to latent factors. As assigning an index to a table can be accomplished by just a single Bernoulli random draw, this new sampling procedure reduces the computational complexity for sampling all $\ell_{vjk}$
from $O(n_{\cdotv\cdotv}K)$ to 
$O\big[\sum_{v}\sum_j\ln(n_{vj}+1)K\big]$, which could lead to a considerable saving in computation 
for long samples where 
large counts $n_{vj}$ are abundant. 

This new sampling algorithm not only is less expensive in computation, but also may converge 
much faster 
as there is no need to worry about the dependencies between the MCMC samples for the factor indices $z_{ji}$, which are not used at all under the compound Poisson representation. Note that the collapsed Gibbs sampler in Section \ref{sec:collapseGibbs} removes the need to sample $\Phimat$ and $\Thetamat$ at the expense of having to sample $z_{ji}$ and augment a $b_{ij}$ for each $z_{ij}$. We will show in Appendix \ref{app:sampler} that maintaining  $\Phimat$ and $\Thetamat$ but removing the need to sample $z_{ji}$ leads to a more effective sampler. 


\subsection{Model comparison}

We also consider NBFA based on the GNBP, in which each of the $J$ samples is assigned with a sample-specific NB process and a globally shared gamma process is mixed with all the $J$ NB processes. Given its connection to DCMLDA, as revealed in Section \ref{sec:3.3}, we call this nonparametric Bayesian model as the GNBP-DCMLDA, 
which is shown to be restrictive in that although each sample has its own factor scores under the corresponding sample-specific factors, all the samples are enforced to have the same factor scores under the same set of globally shared factors. 
 By contrast, by modeling not only the burstiness of the covariates, but also that of the factors, the hGNBP-NBFA provides sample-specific factor scores under the same set of shared factors, making it suitable for extracting low-dimensional latent representations for high-dimensional covariate count vectors.

We describe in detail in Appendix \ref{app:GNBP} how to use the gamma-negative binomial process (GNBP) as a nonparametric Bayesian prior for both PFA and DCMLDA. 
In the prior, for PFA, we have $n_{vj}\sim\mbox{Pois}(\sum_{k}\phi_{vk}\theta_{kj})$, whereas for NBFA, we have $n_{vj}\sim\mbox{NB}(\sum_{k}\phi_{vk}\theta_{kj}, p_j)$, which can be augmented as
$$
\textstyle 
n_{vj}\sim\mbox{Pois}(\lambda_{vj}),~\lambda_{vj}\sim\mbox{Gamma}\left[\sum_k\phi_{vk}\theta_{kj}, {p_j}/({1-p_j})\right].
$$
Thus we have 
$
\textstyle (\lambda_{vj}\,|\,n_{vj}, \Phimat,\thetav_j,p_j) \sim\mbox{Gamma}\left(n_{vj}+\sum_k\phi_{vk}\theta_{kj},~ p_j\right)
$
for NBFA. Similarly, we have $
(\lambda_{vj}\,|-)
\sim\mbox{Gamma}\left(n_{vj}+\sum_k\phi_{vk}r_{k},~ p_j\right)
$ for the GNBP-DCMLDA.
To better understand the similarities and differences between the GNBP-PFA, GNBP-DCMLDA, and hGNBP-NBFA, 
in Table \ref{tab:NBP} we compare their Poisson rates of $n_{vj}$, estimated with the factors and factor scores in a single MCMC sample, and several other important model properties.

\begin{table}[!t]
\begin{footnotesize}
\begin{center}\caption{Comparisons of 
 the GNBP-PFA, GNBP-DCMLDA, and hGNBP-NBFA.}\label{tab:NBP}
\begin{tabular}{ C{7.5pc} | C{6.5pc} C{6.5pc} C{7pc} }
\toprule
 & GNBP-PFA (multinomial mixed-membership model) & GNBP-DCMLDA &hGNBP-NBFA 
(Dirichlet-multinomial mixed-membership model)\\
\midrule
 Estimated Poisson rate of $n_{vj}$ given the factors and factor scores & $~~~~ \sum_k \phi_{vk}\theta_{kj} $ 
 &$(n_{vj}+\sum_k \phi_{vk}r_k)p_j $ 
 & $ (n_{vj}+\sum_k \phi_{vk}\theta_{kj})p_j $ \\
\midrule
Factor analysis 
& $\nv_{j}\sim\mbox{Pois}(\Phimat\thetav_{j})$
 & $\nv_{j}\sim\mbox{NB}(\Phimat\rv, \,p_j)$
 & $\nv_{j}\sim\mbox{NB}(\Phimat\thetav_{j}, \,p_j)$ \\
\midrule
Mixed-membership modeling &  $z_{ji}\sim\mbox{Cat}({\thetav_j}/{\theta_{\cdotv j}})$, $x_{ji}\sim\mbox{Cat}(\phiv_{z_{ji}})$
 &  $\thetav_j\sim\mbox{Dir}(\rv) $, $z_{ji}\sim\mbox{Cat}(\thetav_j)$, $\phiv^{[j]}_k\sim\mbox{Dir}(\phiv_k r_k)$, 
 $x_{ji}\sim\mbox{Cat}(\phiv^{[j]}_{z_{ji}})$ & $\thetav^{[j]}\sim\mbox{Dir}(\thetav_j) $,  $z_{ji}\sim\mbox{Cat}(\thetav^{[j]}),$ $
\phiv^{[j]}_k\sim\mbox{Dir}(\phiv_k \theta_{kj}),$ $x_{ji}\sim\mbox{Cat}(\phiv^{[j]}_{z_{ji}})$ \\
\midrule
Distribution of $X_j = $ $\sum_{k=1}^\infty n_{\cdotv jk} \delta_{\phiv_k}$ given $G$ & $X_j\,|\,G, p_j\sim \mbox{NBP}(G,p_j) $ & $X_j\,|\,G, p_j\sim\mbox{NBP}(G,p_j) $ & $X_j\,|\,G, c_j,p_j\sim\mbox{GNBP}(G,c_j,p_j)$\\
\midrule
Variance-to-mean ratio of $n_{\cdotv jk}$ given $c_j$~and~$p_j$ & $\displaystyle\frac{1}{1-p_j}$& $\displaystyle\frac{1}{1-p_j}$ & $\displaystyle\frac{1}{1-p_j}+\frac{p_j}{c_j(1-p_j)^2}$\\
 \bottomrule
\end{tabular}\vspace{-6mm}
\end{center}
\end{footnotesize}
\end{table}%

To estimate the latent Poisson rates for each count $n_{vj}$ and hence the smoothed normalized covariate frequencies, it is clear from the second row of Table~\ref{tab:NBP} that PFA (the multinomial mixed-membership model) solely relies on the inferred factors $\{\phiv_k\}$ and factor scores $\{\theta_{kj}\}$, DCMLDA adds a sample-invariant smoothing parameter, calculated as $\sum_{v}\phi_{vk}r_k$, into the observed count $n_{vj}$ and weights that sum by a sample-specific probability parameter $p_j$, whereas NBFA (the DMMM model) adds a sample-specific smoothing parameter, calculated as $\sum_{v}\phi_{vk}\theta_{kj}$, into the observed count $n_{vj}$ and weights that sum by $p_j$. Thus PFA represents an extreme that the observed counts are used to infer the factors and factor scores but are not used to directly estimate the Poisson rates; DCMLDA represents another extreme that the covariate frequencies in all samples are indiscriminately smoothed by the same set of smoothing parameters; whereas NBFA combines the observed counts with the inferred sample-specific smoothing parameters. This unique working mechanism also makes NBFA have reduced hyper-parameter sensitivity, as will be demonstrated with experiments. 

\section{Example Results}\label{sec:results}

%

We apply the proposed models to factorize covariate-sample count matrices, 
each column of which
is represented as a $V$ dimensional covariate-frequency count vector, where $V$ is the number of unique covariates. We set the hyper-parameters as $a_0=b_0=0.01$ and $e_0=f_0=1$. 
We consider the JACM ({\footnotesize 
\href{http://www.cs.princeton.edu/~blei/downloads/}{http://www.cs.princeton.edu/$\sim$blei/downloads/}}), 
Psychological Review (PsyReview, {\footnotesize\href{http://psiexp.ss.uci.edu/research/programs_data/toolbox.htm}{http://psiexp.ss.uci.edu/research/programs$\_$data/toolbox.htm}}), 
and NIPS12 ({\footnotesize\href{http://www.cs.nyu.edu/~roweis/data.html}{http://www.cs.nyu.edu/$\sim$roweis/data.html}}) datasets, 
choosing covariates that occur in five or more samples. In addition, we consider the 20newsgroups dataset ({\footnotesize\href{http://qwone.com/~jason/20Newsgroups/}{http://qwone.com/$\sim$jason/20Newsgroups/})}, 
consisting of 18,774 samples from 20 different categories. 
It is partitioned into a training set of 11,269 samples and a testing set of 7,505 ones that were collected at later times. 
We remove a standard list of stopwords and covariates that appear less than five times. 
As summarized in Table \ref{tab:data} of Appendix~\ref{app:plot}, for the PysReview and JACM datasets, each of whose sample corresponds to the abstract of a research paper, the average sample lengths are only about 56 and 127, respectively. 
By contrast, 
a NIPS12 sample that includes the words of all sections of a research paper is in average more than ten times longer. By varying the percentage of covariate {indices} randomly selected from each sample for training, we construct a set of covariate-sample matrices with a large variation on the average lengths of samples, which will be used to help make comparison between different models. 
%
%
%
%
 Depending on applications, we either treat the Dirichlet smoothing parameter $\eta$ as a tuning parameter, or sample it via data augmentation, as described in Appendix \ref{app:eta}. 
 
 To learn the factors in all the following experiments, we use the compound Poisson representation based blocked Gibbs sampler for both the hGNBP-NBFA and GNBP-NBFA, and use collapsed Gibbs sampling for the GNBP-PFA. We compare different samplers for the hGNBP-NBFA and provide justifications for choosing these samplers in Appendix \ref{app:sampler}.

%
%


\subsection{Prediction of heldout covariate {indices} }

\subsubsection{Experimental settings}




We randomly choose 
a certain percentage
of the covariate {indices} in each sample as training, and use the remaining ones to calculate heldout perplexity. As shown in \citet{NBP2012}, the GNBP-PFA performs similarly to the hierarchical Dirichlet process LDA of \citet{HDP} and outperforms a wide array of discrete latent variable models, thus we choose it 
 for comparison. To demonstrate the importance of modeling both the burstiness of the covariates and that of the factors, we also make comparison to the GNBP-DCMLDA that considers only covariate burstiness. 
Since the inferred number of factors and hence the performance 
often depends on the Dirichlet smoothing parameter $\eta$, we set $\eta$ as 
$0.005$, $0.02$, $0.01$, $0.05$, $0.10$, $0.25$, or $0.50$. 
 We vary both the training percentage and $\eta$ to examine how the average sample length and the value of $\eta$ influence the behaviors of the GNBP-PFA, GNBP-DCMLDA, and hGNBP-NBFA and impact their performance relative to each other. 

For all three algorithms, 
we initialize the number of factors as $K=400$ and consider 5000 Gibbs sampling iterations, 
with the first 2500 samples discarded and every sample per
five iterations collected afterwards. 
 For each collected sample, for the GNBP-PFA, we draw the factors $(\phiv_k\,|\,-)\sim\mbox{Dir}(\eta+n_{1\cdotv k},\ldots,\eta+n_{V\cdotv k})$ and factor scores $(\theta_{kj}\,|\,-)\sim\mbox{Gamma}(n_{jk}+r_k,p_j)$ for $k\in\{1,\ldots,K^++K_\star\}$, where we let $n_{v\cdotv k}=0$ for all $k>K^+$; 
for the GNBP-DCMLDA, we draw the factors $(\phiv_k\,|\,-)\sim\mbox{Dir}(\eta+\ell_{1\cdotv k},\ldots,\eta+\ell_{V\cdotv k})$ and the weights $(r_{k}\,|\,-)\sim\mbox{Gamma}(\ell_{\cdotv \cdotv k},1/[c_0-\sum_j\ln(1-p_j)])$, where we let $\ell_{v\cdotv k}=0$ and $\ell_{\cdotv \cdotv k}=\gamma_0/K_{*}$ for all $k>K^+$; 
and for the hGNBP-NBFA, we draw the factors $(\phiv_k\,|\,-)\sim\mbox{Dir}(\eta+\ell_{1\cdotv k},\ldots,\eta+\ell_{V\cdotv k})$ and factor scores $(\theta_{kj}\,|\,-)\sim\mbox{Gamma}[r_k + \ell_{\cdotv j k},1/(c_j-\ln(1-p_j)]$, where we let $\ell_{v\cdotv k}=0$ and $r_k=\gamma_0/K_\star$ for all $k>K^+$. We set $K_{*}=20$ for all three algorithms. 

We compute
the heldout perplexity as
$$
\exp\left(-\frac{1}{m^{\text{test}}_{\cdotv \cdotv}}\sum_{v} \sum_{j} m^{\text{test}}_{vj} \ln \frac{ \sum_{s} \lambda_{vj}^{(s)} }{\sum_{s}\sum_{v'}\lambda_{v'j}^{(s)}} \right)
,$$
where $s\in\{1,\ldots,S\}$ is the index of a collected MCMC sample, $m^{\text{test}}_{vj}$ is the number of test covariate {indices} at covariate $v$ in sample $j$, $m^{\text{test}}_{\cdotv \cdotv}= \sum_v\sum_j m^{\text{test}}_{vj}$, and $\lambda_{vj}^{(s)}$ are computed using the equations shown in the second row of Tabel \ref{tab:NBP}, 
$e.g.$, we have $\lambda^{(s)}_{vj} = \left(n_{vj}+\sum_{k=1}^{ K^++K_{\star}} \phi_{vk}^{(s)} \theta_{kj}^{(s)}\right) p_j^{(s)}$ for the hGNBP-NBFA.
%
For each unique combination of $\eta$ and the training percentage, the results are averaged over five random training/testing partitions. The evaluation method is similar to those used in \citet{
wallach09}, \citet{DILN_BA}, and \citet{NBP2012}. 
All algorithms are coded in MATLAB, with the steps of sampling factor and table indices coded in C to optimize speed. We terminate a trial and omit the results for that particular setting if it takes a single core of an 
Intel Xeon 
3.3 GHz CPU more than 24 hours to finish 5000 iterations. 
The code will be made available in the author's website for reproducible research.

We first consider the NIPS12 dataset, whose average sample length is about 1323, and present its results in Figures \ref{fig:perplexity1_1}-\ref{fig:perplexity1_2_time}. We also consider both the PsyReview and JACM datasets, whose average sample lengths are about 56 and 127, respectively, and provide related plots 
 in 
 Appendix \ref{app:plot}. 


\vspace{-4mm}
\subsubsection{General observations}
 \vspace{-2mm}
For multinomial mixed-membership models, generally speaking, the smaller the Dirichlet smoothing parameter $\eta$ is, the more sparse and specific the inferred factors are encouraged to be, and the larger the number of inferred factors using a nonparametric Bayesian mixed-membership modeling prior, such as the hierarchical Dirichlet process and the gamma- and beta-negative binomial processes 
\citep{DILN_BA,BNBP_PFA_AISTATS2012}. As shown in Figures \ref{fig:perplexity1_1}(a)-(e), for the hGNBP-NBFA, a nonparametric Bayesian DMMM model, we observe a relationship between the 
 number of inferred factors and $\eta$ similar to that for the GNBP-PFA, a nonparametric Bayesian multinomial mixed-membership model. 
 
In comparison to multinomial mixed-membership models such as the GNBP-PFA, 
what make the hGNBP-NBFA different and desirable are: 1) its parsimonious representation that uses fewer factors to achieve better heldout prediction, as shown in Figures \ref{fig:perplexity1_1_time}(a)-(e);
2)
ts distinct mechanism in adjusting 
 its number of inferred factors according to the lengths of samples, as shown in Figures \ref{fig:perplexity1_2}(a)-(f); 
3) its significantly lower computational complexity for a covariate-sample matrix with long samples (large column sums), with the differences becoming increasingly more significant as the the average sample length increases, as shown in Figures \ref{fig:perplexity1_1_time}(f)-(j) and \ref{fig:perplexity1_2_time}(a)-(f);
4) its ability to achieve the same predictive power with significantly less time, as shown in Figures \ref{fig:perplexity1_2_time}(g)-(l);
5) and its overall better predictive performance both under various values of $\eta$ while controlling the sample lengths, as shown in Figures \ref{fig:perplexity1_1}(f)-(j), and under various sample lengths while controlling $\eta$, as shown in Figures \ref{fig:perplexity1_2}(g)-(l).

%
 
\begin{figure}[!tb]
\begin{center}
\includegraphics[width=0.95\columnwidth]{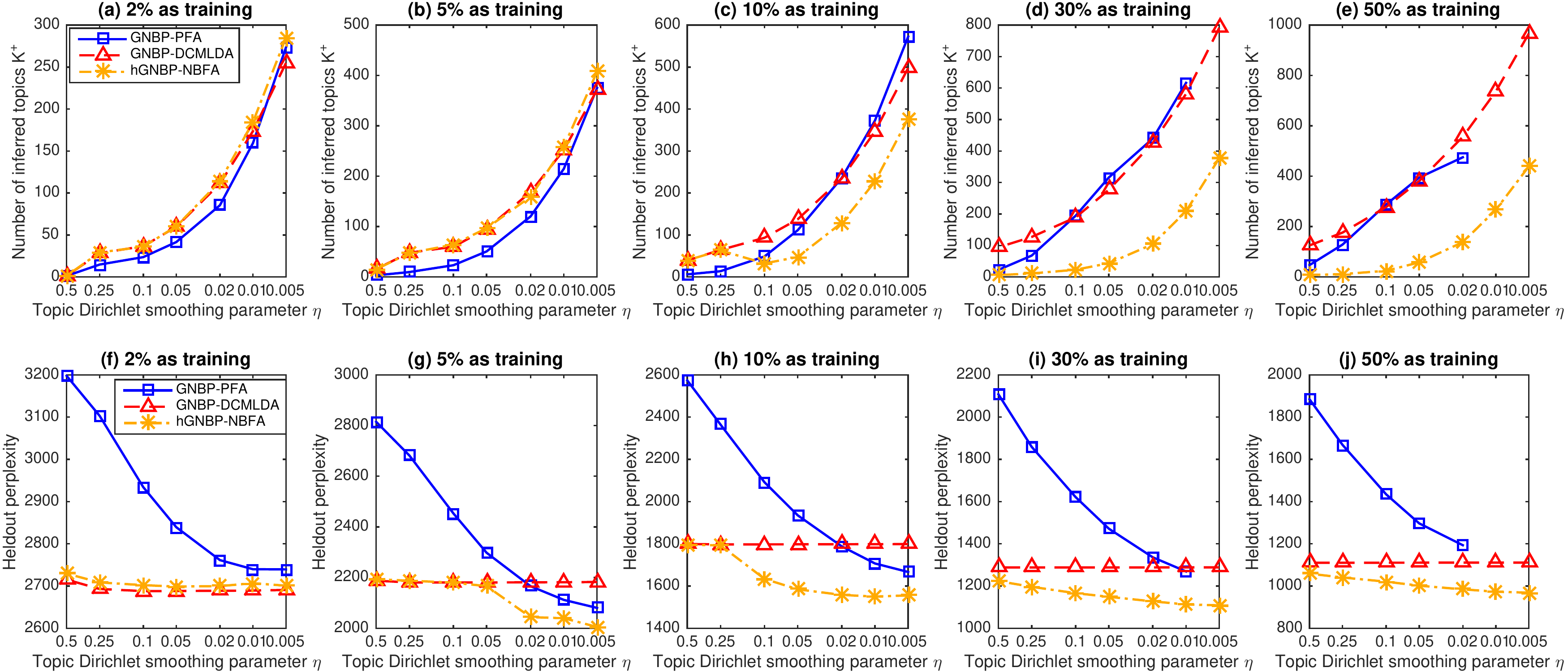}
\end{center}
\vspace{-6.0mm}
\caption{\small \label{fig:perplexity1_1}
Comparisons of the GNBP-PFA (multinomial mixed-membership model), GNBP-DCMLDA, and hGNBP-NBFA (Dirichlet-multinomial mixed-membership model) on (a)-(e) the posterior means of the number of inferred factors $K^+$ and (f)-(j) heldout perplexity, both as a function of the Dirichlet smoothing parameter $\eta$ for the NIPS12 dataset. The values of $\eta$ are plot in the logarithmic scale from large to small. In both rows, the plots from left to right are obtained using $2\%$, $5\%$, $10\%$, $30\%$, and $50\%$ of the covariate {indices} for training, respectively. 
All plots are based on five independent random trials. The error bars are not shown as variations across different trials are small. Some results for the GNBP-PFA are missing as they took more than 24 hours to run 5000 Gibbs sampling iterations on a 3.3 GHz CPU and hence were terminated before completion.
}


\begin{center}
\includegraphics[width=0.95\columnwidth]{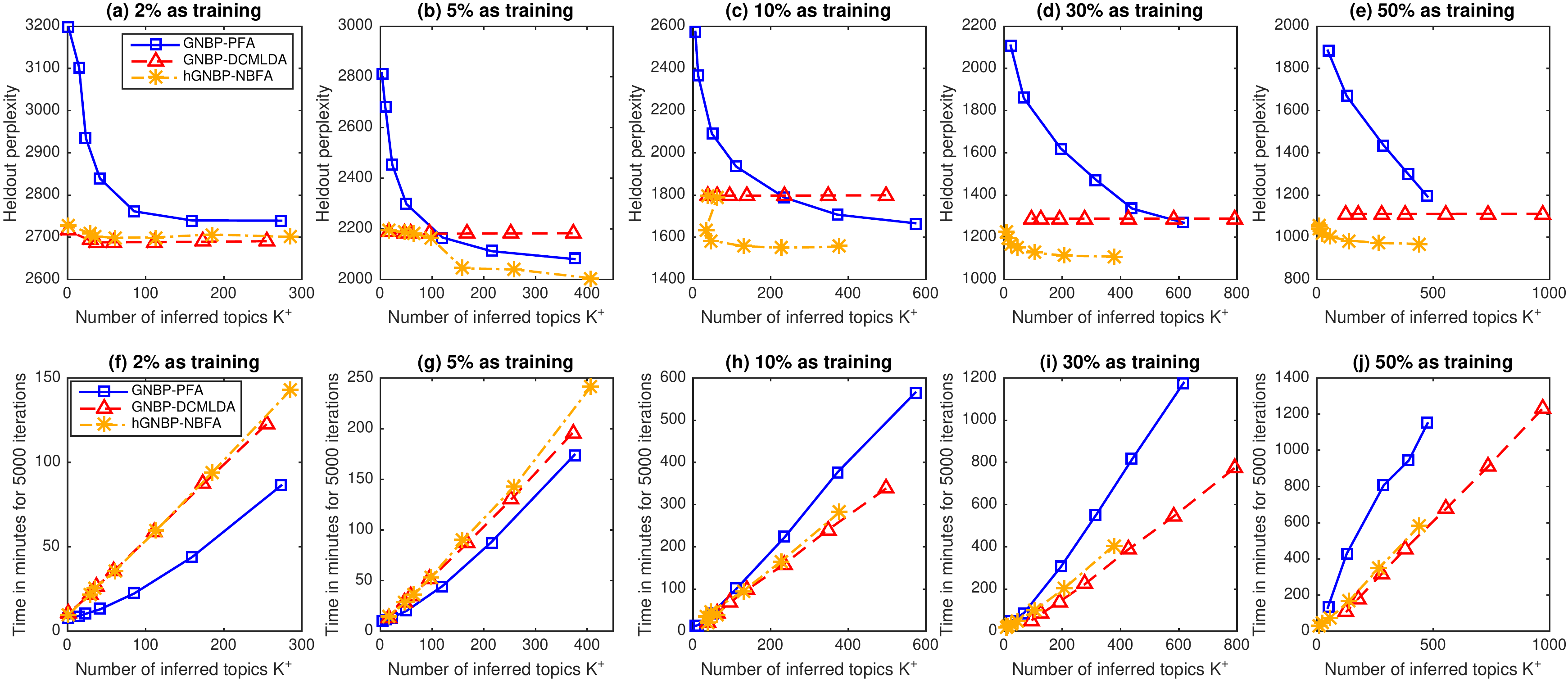}
\end{center}
\vspace{-6.0mm}
\caption{\small \label{fig:perplexity1_1_time}
Using the same results shown in Figure \ref{fig:perplexity1_1}, we plot (a)-(e) the obtained heldout perplexity and (f)-(j) the number of minutes to finish 5000 Gibbs sampling iterations, both as a function of the number of inferred factors $K^+$. 
}
\end{figure}

\begin{figure}[!tb]
\begin{center}
\includegraphics[width=.95\columnwidth]{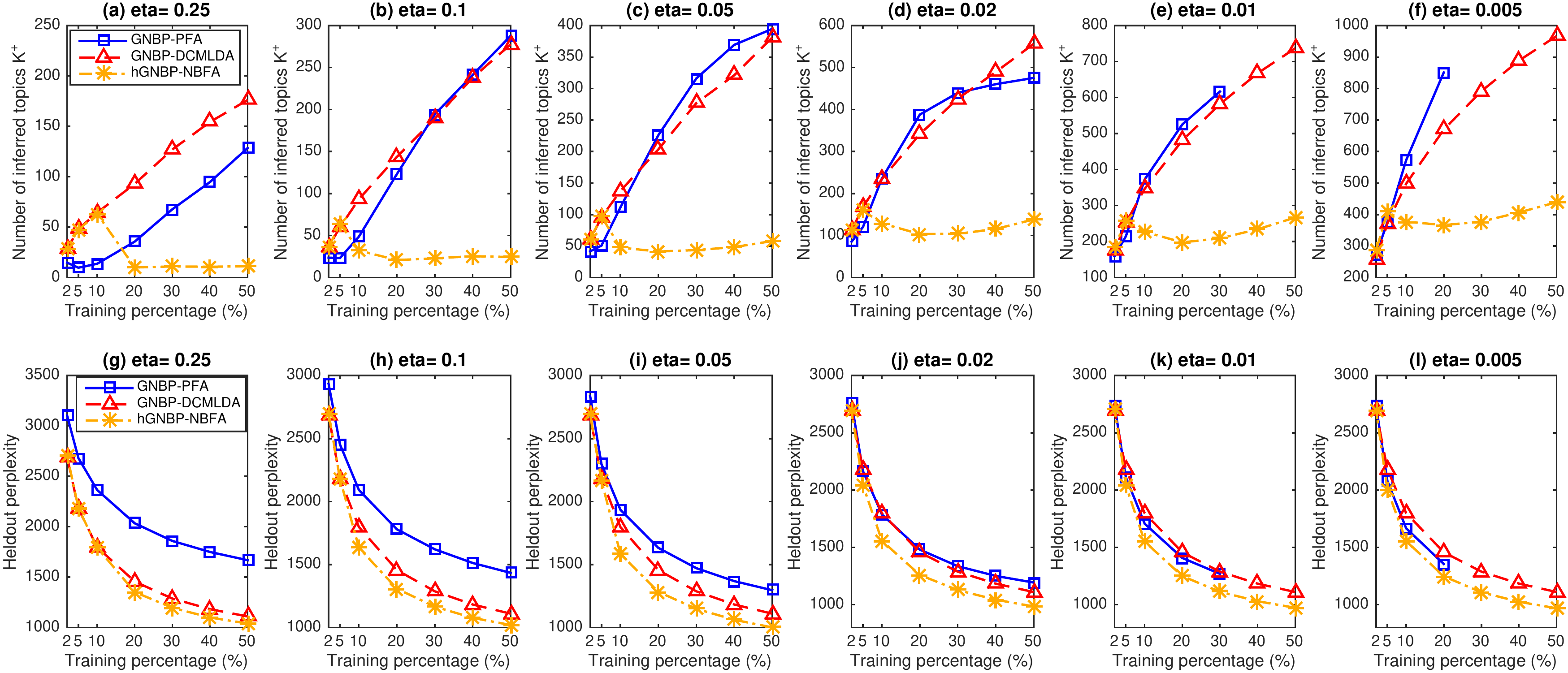}
\end{center}
\vspace{-6.0mm}
\caption{\small \label{fig:perplexity1_2}
Comparisons of the GNBP-PFA (multinomial mixed-membership model), GNBP-DCMLDA, and hGNBP-NBFA (Dirichlet-multinomial mixed-membership model) on (a)-(f) the posterior means of the number of inferred factors $K^+$ and (g)-(l) heldout perplexity, both as a function of the percentage of covariate {indices} used for training for the NIPS12 dataset. In both rows, the plots from left to right are obtained with $\eta=0.25, 0.1,0.05, 0.02, 0.01$, and $0.005$, respectively. Other specifications are the same as those of Figure \ref{fig:perplexity1_1}.
}

\begin{center}
\includegraphics[width=.95\columnwidth]{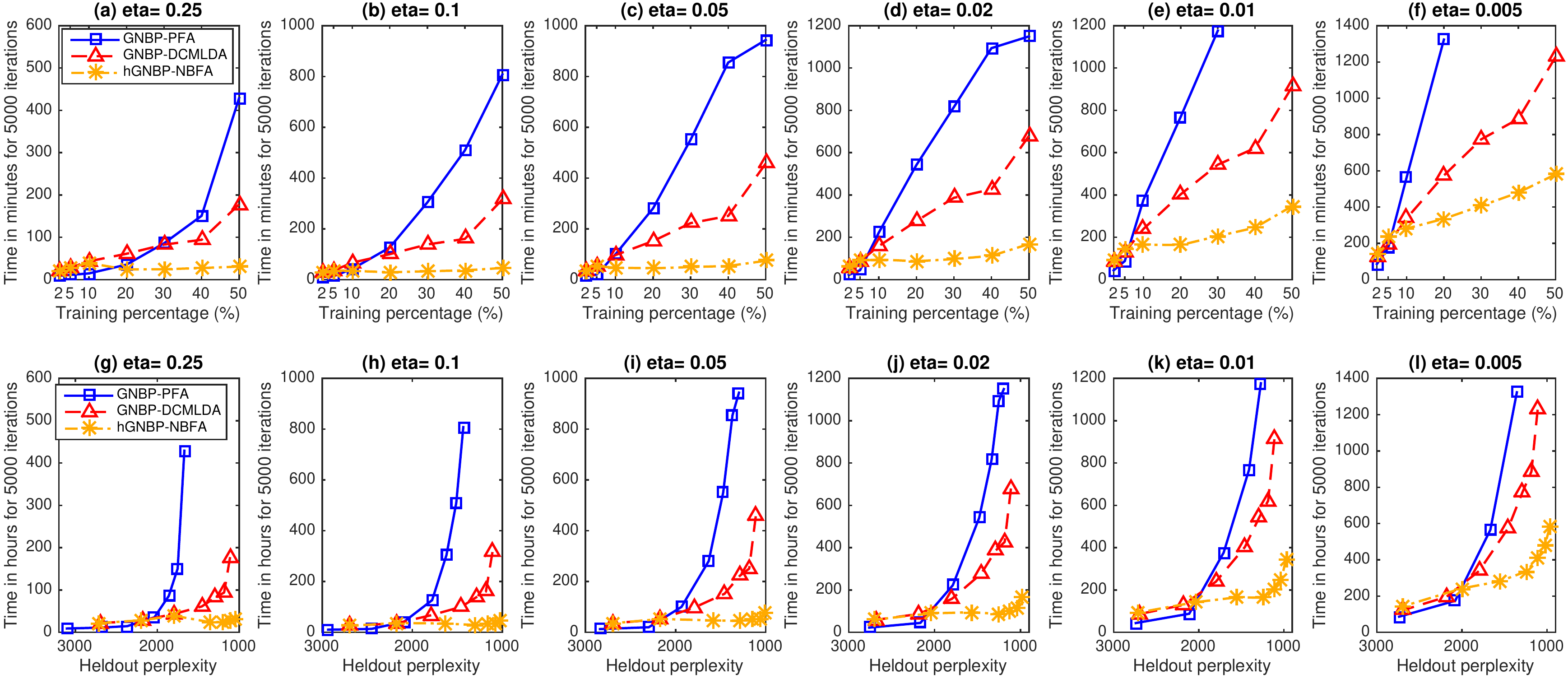}
\end{center}
\vspace{-6.0mm}
\caption{\small \label{fig:perplexity1_2_time}
Using the results shown in Figure \ref{fig:perplexity1_2}, we plot the number of minutes to finish 5000 Gibbs sampling iterations both (a)-(f) as a function of the percentage of covariates used for training and (g)-(l) as a function of heldout perplexity. 
}
\end{figure}


 \subsubsection{Detailed discussions}

\textbf{Distinct behavior and parsimonious representation.}
When fixing $\eta$ but gradually increasing the average sample length, the number of factors inferred by a nonparametric Bayesian multinomial mixed-membership model such as the GNBP-PFA often increases at a near-constant rate, as shown with the blue curves in Figure \ref{fig:perplexity1_2}(a)-(f). The GNBP-DCMLDA, which models covariate burtiness, behaves similarly in the number of inferred factors, as shown with the red curves in Figure \ref{fig:perplexity1_2}(a)-(f). 
Under the same setting, 
the number of inferred factors by the hGNBP-NBFA often first increases at a similar near-constant rate 
when the average sample length is short, however, it starts decreasing once the average sample length becomes sufficiently long, and eventually turns around and increases, but at a much lower rate, as the average sample length further increases, as shown with the yellow curves in Figure \ref{fig:perplexity1_2}(a)-(f). This distinct behavior implies that by exploiting its ability to model both covariate and factor burstiness, the hGNBP-NBFA could have a parsimonious representation of a dataset with long samples. By contrast, a nonparametric Bayesian multinomial mixed-membership model, such as the GNBP-PFA, models neither covariate nor factor burstiness. Consequently, it has to increase its latent dimension at a near-constant rate as a function of the average sample length, in order to adequately capture self- and cross-excitation of covariate frequencies, which are often more prominent in longer samples. 
It is clear that by decreasing $\eta$ and hence increasing the number of inferred factors, the GNBP-PFA can gradually approach and eventually outperform the GNBP-DCMLDA, but still clearly underperform the hGNBP-NBFA in most cases, even if using many more factors and consequently significantly more computation. 


 



\textbf{Combining factorization and the modeling of burstinss.}
As shown in Figure \ref{fig:perplexity1_1}(f), when the training percentage is as small as $2\%$, the GNBP-DCMLDA, which combines the observed counts $n_{vj}$ and the inferred sample-invariant smoothing parameters $\sum_k \phi_{vk}r_k $ to estimate the Poisson rates (and hence the smoothed normalized covariate frequencies), achieves the best predictive performance (lowest heldout perplexity); the hGNBP-NBFA tries to improve DCMLDA by combining the observed counts and document-specific smoothing parameters $\sum_k\phi_{vk}\theta_{kj}$, and the GNBP-PFA only relies on $\sum_k\phi_{vk}\theta_{kj}$, yielding slightly and significantly worse performance, respectively, at this relatively extreme setting. 
This suggests that when the observed counts are too small, using factorization may not provide any advantages than simply smoothing the raw covariate counts with sample-invariant smoothing parameters. 

As the training percentage increases, all three algorithms quickly improve their performance, as shown in Figures \ref{fig:perplexity1_1}(g)-(j). Given a training percentage that is sufficiently large, e.g., 10\% for this dataset ($i.e.$, the average training sample length is about 132), all three algorithms tend to increase their numbers of inferred factors $K^+$ as $\eta$ decreases, although the hGNBP-NBFA usually has a lower increasing rate. They differ from each other significantly, however, on how the performance improves as the inferred number factors increases, as shown in Figures \ref{fig:perplexity1_1_time}(a)-(e): for the GNBP-DCMLDA, as it relies on $\sum_k \phi_{vk}r_k $ to smooth the observed counts, its predictive power is almost invariant to the change of $\eta$ and its number of factors; for the GNBP-PFA, by decreasing $\eta$ and hence increasing its number of inferred factors, it can approach and eventually outperform DCMLDA; whereas for the hGNBP-NBFA, it follows DCMLDA closely when $\eta$ is large or the lengths of samples are short, but often reduces its rate of increase for the number of inferred factors as $\eta$ decreases and quickly lowers its perplexity as $K^+$ increases, as long as $\eta$ is sufficiently small or the samples are sufficiently long. Thus in general, 
the hGNBP-NBFA provides the lowest perplexity using the least number of inferred factors.

Note that when the lengths of the training samples are short, setting $\eta$ to be large will make the factors $\phiv_k$ of NBFA become over-smoothed and hence NBFA becomes essentially the same as DCMLDA. As $\eta$ decreases given the same average sample length, or as the average sample length increases given the same $\eta$, the factorization of NBFA with sample-dependent factor scores 
gradually take effect to improve the estimation of the Poisson rates and hence the smoothed normalized covariate frequencies for each sample. Overall, by combining the factorization, as used in PFA, the modeling of covariate burstiness, as used in DCMLDA, and the modeling of factor burstiness, unique to NBFA, the hGNBP-NBFA captures both self- and cross-excitation of covariate frequencies and achieves the best predictive performance with the most parsimonious representation as long as the average sample length is not too short and the value of $\eta$ is not set too large to overly smooth the~factors. 

%

\textbf{Significantly lower computation for sufficiently long samples.}
For the GNBP-PFA, the collapsed Gibbs sampler samples all the factor indices with a computational complexity of $O(n_{\cdotv \cdotv} K^+)$, whereas for the hGNBP-NBFA, the corresponding computation has a complexity of $ O\big[\sum_{v}\sum_j\ln(n_{vj}+1) K^+\big]$ and sampling $\{\phi_k\}_k$ and $\{\thetav_j\}_j$ adds an additional computation of $O(VK^++NK^+)$. Thus the computation for the hGNBP-NBFA not only is often lower given the same $K^+$ for a dataset consisting of sufficiently long samples, but also becomes much lower because the inferred $K^+$ is often much smaller when the sample lengths are sufficiently long. For example, as shown in Figure \ref{fig:perplexity1_1_time}(i), when 30\% of the covariate {indices} in each sample are used for training, which means the average training sample length is about 397, the time for the GNBP-PFA to finish 5000 Gibbs sampling iteration on a 3.3 GHz CPU is about double that for the hGNBP-NBFA when their inferred numbers of factors are similar to each other; and when 20\% of the covariate {indices} in each sample are used for training ($i.e.$, the average training sample length is around 265), in comparison to the hGNBP-NBFA, the GNBP-PFA takes about three times more minutes when $\eta=0.1$, as shown in Figures \ref{fig:perplexity1_2_time}(b), and four times more minutes when $\eta=0.01$, as shown in Figures \ref{fig:perplexity1_2_time}(e).
Overall, for a dataset whose samples are not too short to exhibit self- and cross-excitation of covariate frequencies, the hGNBP-NBFA often takes the least time to finish the computation while controlling the value of $\eta$ and average sample length, has lower computation given the same inferred number of factors~$K^+$, and achieves a low perplexity with significantly less computation.

\subsection{Unsupervised feature learning for classification}

%



To further verify the advantages of NBFA that models both self- and cross-excitation of covariate frequencies, 
 we use the proposed models to extract low-dimensional feature vectors from high-dimensional covariate-frequency count vectors 
of the 20newsgroups dataset, and then examine how well the unsupervisedly extracted feature vector of a test sample can be used to correctly classify it to one of the 20 news groups. As the classification accuracy often strongly depends on the dimension of the feature vectors, we truncate the total number of factors at $K=25$, 50, 100, 200, 400, 600, 800, or 1000. Correspondingly, we slightly modify the gamma process based nonparametric Bayesian models by choosing a discrete base measure for the gamma process as $G_0=\sum_{k=1}^K \frac{\gamma_0}K \delta_{\phiv_k},~\phiv_k\sim\mbox{Dir}(\eta,\ldots,\eta)$. Thus in the prior we now have $r_k=G(\phiv_k)\sim\mbox{Gamma}(\gamma_0/K,1/c_0)$ and consequently the Gibbs sampling update equations for $\{r_k\}_k$ and $\gamma_0$ will also slightly change. We omit these details for brevity, and refer 
 to \citet{NBP2012} on how the same type of finite truncation is used in inference for nonparametric Bayesian models.

For this application, we fix the truncation level $K$ but impose the non-informative $\mbox{Gamma}(0.01,1/0.01)$ prior on the Dirichlet smoothing parameter $\eta$, letting it be inferred from the data using \eqref{eq:sampleeta}. 
The same as before, we consider collapsed Gibbs sampling for the GNBP-PFA and the compound Poisson representation based blocked Gibbs sampler for the hGNBP-NBFA, with the main difference in that a fixed instead of an adaptive truncation is now used for inference. We do not consider the GNBP-DCMLDA here since it does not provide sample specific feature vectors under the same set of shared factors. Note that although we fix $K$, if $K$ is set to be large enough, not necessarily all factors would be used and hence a truncated model still preserves its ability to infer the number of active factors $K^+\le K $; whereas if $K$ is set to be small, a truncated model may lose its ability to infer $K^+$, but it still maintains asymmetric priors 
\citep{wallach2009rethinking} on the factor scores. 

For both the hGNBP-NBFA and GNBP-PFA, we consider 2000 Gibbs sampling iterations on the 11,269 training samples of the 20newsgroups dataset, and retain the weights $\{r_k\}_{1,K}$ and the posterior means of $\{\phiv_k\}_{1,K}$ as factors, according to the last MCMC sample, for testing. With these $K$ inferred factors and weights, we further apply 1000 blocked Gibbs sampling iterations for both models and collect the last 500 MCMC samples to estimate the posterior mean of the feature usage proportion vector $\thetav_j /\theta_{\cdotv j}$, for every sample in both the training and testing sets. Denote $\bar{\thetav}_j\in \mathbb{R}^{K}$ as the estimated feature vector for sample~$j$. We use the $L_2$ regularized logistic regression provided by the LIBLINEAR ({\footnotesize\href{https://www.csie.ntu.edu.tw/~cjlin/liblinear/}{https://www.csie.ntu.edu.tw/$\sim$cjlin/liblinear/})} package \citep{REF08a} to train a linear classifier on all $\bar{\thetav}_j$ in the training set and use it to classify each $\bar{\thetav}_j$ in the test set to one of the 20 news groups; the regularization parameter $C$ of the classifier five-fold cross validated on the training set from $(2^{-10}, 2^{-9},\ldots, 2^{15})$. 

\begin{figure}[!tb]
\begin{center}
\includegraphics[width=0.70\columnwidth]{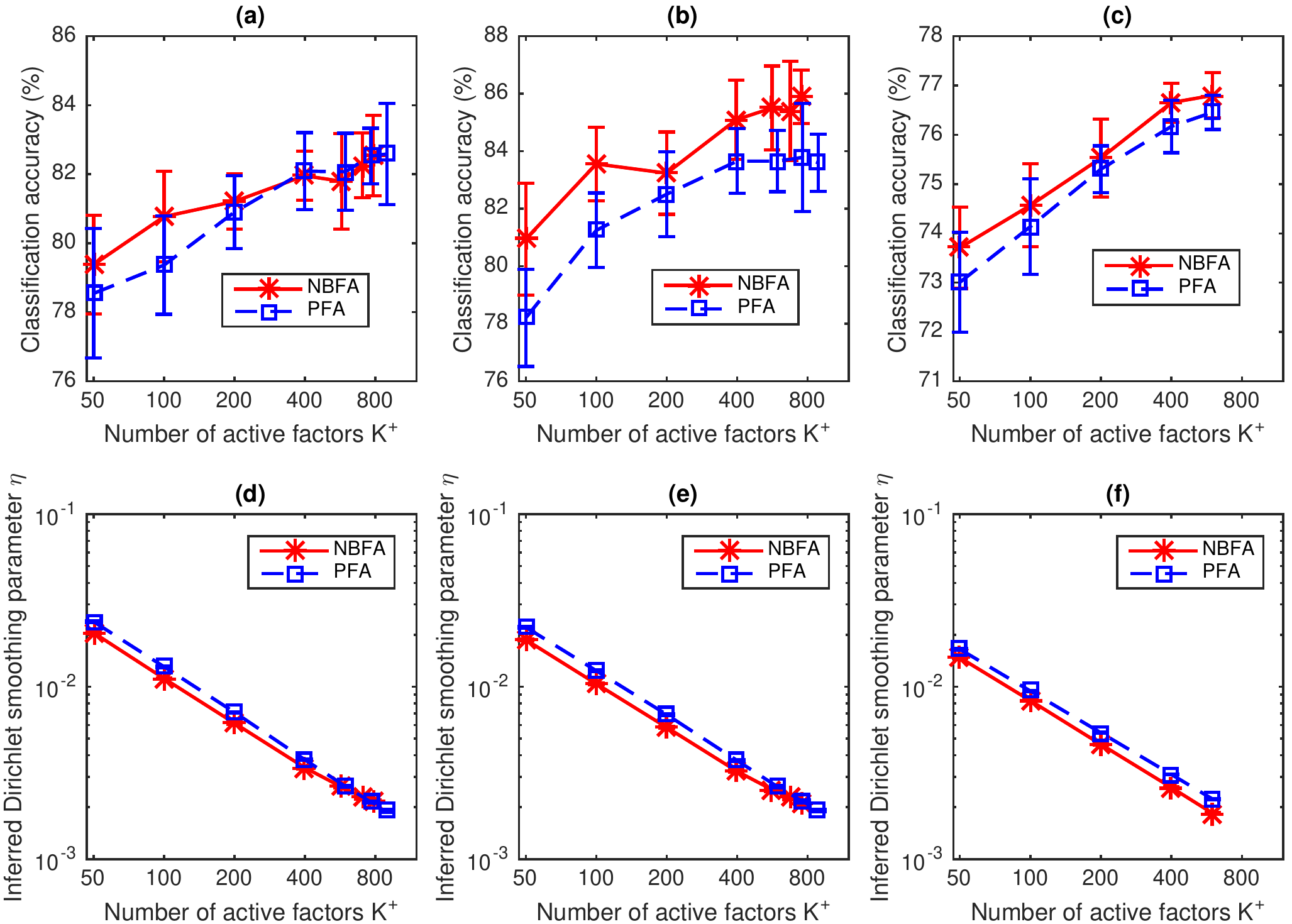}
\end{center}
\vspace{-5.5mm}
\caption{\small \label{fig:classify}
Comparison between negative binomial factor analysis (NBFA) and Poisson factor analysis (PFA) on three different classification tasks, with the number of factors $K$ fixed and the Dirichlet smoothing parameter $\eta$ inferred from the data. 
 For the $alt.atheism$ versus $talk.religion.misc$ binary classification task, based on twelve independent random trials, we plot (a) the classification accuracy and (d) the inferred $\eta$, both as a function of the number active topics $K^+\le K$, where $K\in\{50, 100, 200, 400, 600, 800, 1000\}$.
 (d) and (e): Analogous plots to
(a) and (b) for the $comp.sys.ibm.pc.hardware$ versus $comp.sys.mac.hardware$ binary classification task, with $K\in\{50, 100, 200, 400, 600, 800, 1000\}$. 
(c) and (f): Analogous plots to (a) and (b) for the 20newsgroups multi-class classification task, with $K\in\{50, 100, 200, 400, 600\}$. 
 \vspace{-5mm}
}
\end{figure}

We first consider 
distinguishing between 
the $alt.atheism$ and $talk.religion.misc$ news groups, 
and between the $comp.sys.ibm.pc.hardware$ and $comp.sys.mac.hardware$ 
news groups. 
For each binary classification 
task, we remove a standard list of stop words and only consider the covariates that appear at least five times in both newsgroups combined, and report the classification accuracies based on twelve independent runs with random initializations. For the first binary classification task, we have  856 training documents, with 6509 unique terms and about 116K words,  while  for the second one, we have 1162  training documents, with 4737  unique terms and about 91K words.

As shown in Figures \ref{fig:classify}(a) and \ref{fig:classify}(b), NBFA clearly outperforms PFA for both binary classification tasks in that it in general provides higher classification accuracies on testing samples while controlling the 
truncation level $K$ ($i.e.$, the dimension of the extracted feature vectors). 
It is also interesting to examine how the inferred Dirichlet smoothing parameter $\eta$ changes as the truncation level $K$ increases, as shown in Figures \ref{fig:classify}(d) and \ref{fig:classify}(e). It appears that the inferred $\eta$'s and 
active factors $K^+$'s could be fitted with a decreasing straight line in the logarithmic scale, except for the tails that seem slightly concave up, for both NBFA and PFA. 
When the truncation level $K$ is not sufficiently large, the inferred $\eta$ of NBFA is usually smaller than that of PFA given the same $K$. This may be explained by examining (\ref{eq:sampleeta}), where $\ell_{vjk}\le n_{vjk}$ a.s. and the differences could be significant for large~$n_{vjk}$. Note when using the raw word counts as the features, 
the classification accuracies of logistic regression are 79.4\% and 86.7\% for the first and second binary classification tasks, respectively, and when using  the normalized term frequencies as the features, these are 80.8\% and 88.0\%,  respectively.

%
%
%

%
%

In addition to these two binary classification tasks, we consider multi-class classification on the 20newsgroups dataset. After removing stopwords and terms (covariates) that appear less than five times, we obtain $ 33,420$ unique terms and about 1.4 million words for training, as summarized in Table \ref{tab:data}. We use all 11,269 training documents to infer the factors and factor scores, and mimic the same testing procedure used for binary classification to extract low-dimensional feature vectors, with which each testing sample is classified to one of the 20 news groups using the same $L_2$ regularized \ logistic regression. Note 
the classification accuracies of logistic regression with the raw counts or normalized term frequencies as features are 78.0\% and 79.4\%, respectively. As shown in Figure \ref{fig:classify}(f), NBFA generally outperforms PFA in terms of classification accuracies given the same feature dimensions, consistent with our observations for both binary classification tasks. We also observe similar relationship between the $K^+$ and inferred $\eta$ 
as we do in both binary classification tasks.

\section{Conclusions}\label{sec:conclusion}

Negative binomial factor analysis (NBFA) is proposed to factorize the covariate-sample count matrix under the NB likelihood. Its equivalent representation as the Dirichlet multinomial mixed-membership model reveals its distinctions 
from previously proposed discrete latent variable models. The hierarchical gamma-negative binomial process (hGNBP) is further proposed to support NBFA with countably infinite factors, and a compound Poisson representation based blocked Gibbs sampler 
 is shown to 
converge fast and have low computational complexity. By capturing both self- and cross-excitation of covariate frequencies and by smoothing the observed counts with both sample and covariate specific rates obtained through factorization under the NB likelihood, 
 the hGNBP-NBFA not only infers a parsimonious representation of a covariate-sample count matrix, but also 
 achieves state-of-the-art predictive performance at low computational cost. In addition, the latent feature vectors inferred under the hGNBP-NBFA are better suited for classification than those inferred by the GNBP-PFA. It is of interest to investigate a wide variety of 
 extensions built on Poisson factor analysis 
 under this new modeling framework.




\bibliographystyle{ba}
\bibliography{
References052016_1.bib}

\begin{acknowledgement}
 The author would like to thank the editor-in-chief, editor, associate editor, and two anonymous referees for their 
invaluable comments and suggestions, which have helped improve the paper
substantially.  
\end{acknowledgement}

\newpage

\appendix
\begin{center}
\Large{\textbf{Nonparametric Bayesian Negative Binomial Factor Analysis: Supplementary Material}}
\end{center}
\renewcommand{\thechapter}{A}

\renewcommand{\thesection}{\Alph{section}}

\numberwithin{equation}{section}
\section{Proofs}\label{app:proof}
\begin{proof}[Proof of Theorem 1]
With $\tv:=(t_0,\ldots,t_K)\in \mathbb{R}^{K+1}$, the characteristic function of $\xv$ can be expressed as
\beq
\E\left[e^{i\tv^T\xv}\right] = \prod_{k=1}^K \E\left[e^{i(t_0+t_k)x_k}\right] =\prod_{k=1}^K\left(\frac{1-p}{1-pe^{i(t_0+t_k)}}\right)^{r_k}.\notag
\eeq
We augment $\yv$ as
\beq
(y_1,\ldots,y_K)\sim\mbox{Mult}(y,\thetav),~\thetav\sim\mbox{Dir}(r_1,\ldots,r_K),~y\sim\mbox{Pois}(\lambda),~\lambda\sim\mbox{Gamma}(r_{\cdotv},p/(1-p)), \notag
\eeq
where $\thetav=(\theta_1,\ldots,\theta_K)^T$.
Conditioning on $\thetav$ and $y$, we have
$$
\E\left[e^{i\tv^T\yv}\,|\,\thetav,y\right] = \left(\sum_{k=1}^K\theta_k e^{i(t_0+t_k)}\right)^y;
$$
conditioning on $\thetav$ and $\lambda$, we have
\begin{align}
\E\left[e^{i\tv^T\yv}\,|\,\thetav,\lambda\right] &= \E_y\left[\left(\sum_{k=1}^K\theta_k e^{i(t_0+t_k)}\right)^y\right]\notag\\
& = \exp\left[\lambda\left(\sum_{k=1}^K\theta_k e^{i(t_0+t_k)}-1\right)\right]\notag\\
&=\prod_{k=1}^Ke^{\lambda_k \left(e^{i(t_0+t_k)}-1\right)}, \notag
\end{align}
where $\lambda_k = \lambda \theta_k$ are independent gamma random variables, as the independent product of the gamma random variable $\lambda$ and the Dirichlet random vector $\thetav$, with the gamma shape parameter and Dirichlet concentration parameter both equal to $r_{\cdotv}$, leads to independent gamma random variables. 
Further marginalizing out $\lambda_k$, we have
\beq
\E\left[e^{i\tv^T\yv}\right] = \left[1-\frac{p}{1-p}\left(e^{i(t_0+t_k)}-1\right)\right]^{-r_k} =\prod_{k=1}^K\left(\frac{1-p}{1-pe^{i(t_0+t_k)}}\right)^{r_k}. \notag
\eeq
Thus $\xv$ and $\yv$ are equal in distribution as their characteristic functions are the same. 
\end{proof}

\begin{proof}[Poof of Proposition \ref{lem:DirMult}]
Multiplying the likelihood in (\ref{eq:DirMult_like}) with the PMF of the NB distribution in (\ref{eq:NBPF4}), we have
\beq
P(\xv_j,\zv_j,n_j\,|\,\Phimat,\thetav_j,p_j)=\frac{1}{n_j!}\prod_{v=1}^V\prod_{k=1}^K \frac{\Gamma(n_{vjk}+\phi_{vk}\theta_{kj})}{\Gamma(\phi_{vk}\theta_{kj})}p_j^{n_{vjk}}(1-p_j)^{\phi_{vk}\theta_{kj}}, \notag 
\eeq
which, multiplied by the combinatorial coefficient ${n_j!}/({\prod_{v=1}^V\prod_{k=1}^K n_{vjk}!})$, becomes the same as (\ref{eq:NBPF5}).
\end{proof}

\begin{proof}[Proof of Proposition \ref{lem:CRP_DirCat}]
For the first hierarchical model, we have
\begin{align}
P(\bv,\zv\,|\,n,\rv,p) &= \left\{\prod_{k=1}^{K} \mbox{CRP}(\{b_i\}_{i:z_i=k}; n_{k},r_k)\right\}\mbox{DirCat}(\zv;n,\rv)\notag\\
& = \frac{\Gamma(r_{\cdotv})}{\Gamma(n+r_{\cdotv})}\prod_{k=1}^{K} \mbox{CRP}(\{b_i\}_{i:z_i=k}; n_{k},r_k)\frac{\Gamma(n_k+r_k)}{\Gamma(r_k)} \notag\\
& =\frac{\Gamma(r_{\cdotv})}{\Gamma(n+r_{\cdotv})}\prod_{k=1}^{K} \left\{ r_k^{\ell_k}\prod_{t=1}^{\ell_k}\Gamma(n_{kt}) \right\}, \notag
\end{align}
where $\ell_k$ is the number of unique indices in $\{b_i\}_{i:z_i=k}$ and $n_{kt}= \sum_{i=1}^n \delta(b_i=t,z_i=k)$.

For the second hierarchical model, we have
\begin{align}
P(\bv,\zv\,|\,n,\rv,p) &= P(\zv\,|\,\bv,\rv) \mbox{CRP}(\bv;n,r_{\cdotv})\notag\\
& = P(\zv\,|\,\bv,\rv) r_{\cdotv}^\ell \frac{\Gamma(r_{\cdotv})}{\Gamma(n+r_{\cdotv})} \prod_{t=1}^{\ell}\Gamma\left(\sum_{i=1}^n\delta(b_i=t)\right)\notag\\
& = \frac{\Gamma(r_{\cdotv})}{\Gamma(n+r_{\cdotv})} \left\{\prod_{k=1}^{K} r_k^{\sum_{t=1}^\ell \delta(s_t=k)} \right\}\left\{\prod_{t=1}^{\ell}\Gamma\left(\sum_{i=1}^n\delta(b_i=t)\right)\right\}\notag\\
& =\frac{\Gamma(r_{\cdotv})}{\Gamma(n+r_{\cdotv})}\prod_{k=1}^{K} \left\{ r_k^{\ell_k}\prod_{t=1}^{\ell_k}\Gamma(n_{kt}) \right\}, \notag
\end{align}
where $\ell_k =\sum_{t=1}^\ell \delta(s_t=k) $ is the number of unique indicies in $\{b_i\}_{i:z_i=k}$ and $n_{kt} =\delta(s_t=k) \sum_{i=1}^n \delta(b_i=t)= \sum_{i=1}^n \delta(b_i=t,z_i=k)$.

Simply applying the chain rule, we have
\begin{align}
P(\bv,\zv,n\,|\,\rv,p) &= P(\bv,\zv\,|\,n,\rv) \mbox{NB}(n;r_{\cdotv},p) = \frac{1}{n!}\prod_{k=1}^{K} \left\{ r_k^{\ell_k}(1-p)^{r_k}p^{n_k}\prod_{t=1}^{\ell_k}\Gamma(n_{kt}) \right\} . \notag
\end{align}
The mapping from $\{\bv,\zv,n\}$ to $\{\lv,\nv,n\}$ is many to one, with
\begin{align}
P(\lv,\nv,n\,|\,\rv,p) 
& = \prod_{k=1}^K \left\{\sum_{(n_1,\ldots,n_k)\in\mathcal{D}_{n_k,\ell_k}}\frac{1}{\ell_k! \prod_{t=1}^{\ell_k}n_{kt}!} r_k^{\ell_k}(1-p)^{r_k}p^{n_k}\prod_{t=1}^{\ell_k}\Gamma(n_{kt}) \right\}\notag\\
&=
\prod_{k=1}^{K} \left\{ r_k^{\ell_k}(1-p)^{r_k}p^{n_k} \frac{\,|\,s(n_k,\ell_k)\,|\,}{n_k!}\right\}, \notag
\end{align}
where $\mathcal{D}_{n_k,\ell_k}:=\{(n_{k1},\ldots,n_{k\ell_k}): n_{kt}\ge 1\text{ and }\sum_{t=1}^{\ell_k} n_{kt} = n_k\}$.
For \eqref{eq:PoissonLogBig}, using the PMF of the Poisson-logarithmic bivariate distribution shown in (\ref{eq:Po-log}), we have
\begin{align}
P(\lv,\nv,n\,|\,\rv,p) 
& = \prod_{k=1}^{K} \left\{ r_k^{\ell_k}(1-p)^{r_k}p^{n_k} \frac{\,|\,s(n_k,\ell_k)\,|\,}{n_k!}\right\}. \notag
\end{align}
\end{proof}

\begin{proof}[Proof of Theorem \ref{cor:CRTMult}]
For the first hierarchical model, we have
\begin{align}
P(\nv, \ellv\,|\,n,\rv) &= \left\{\prod_{k=1}^{K} \mbox{CRT}(\ell_k; n_{k},r_k)\right\}\mbox{DirMult}(\nv;n,\rv)\notag\\
& =\frac{n!}{\prod_{k=1}^K n_k!}\frac{\Gamma(r_{\cdotv})}{\Gamma(n+r_{\cdotv})}\prod_{k=1}^{K} r_k^{\ell_k}|s(n_k,\ell_k)|\notag\\
& = \frac{n!}{\prod_{k=1}^K \ell_k!}\frac{\Gamma(r_{\cdotv})}{\Gamma(n+r_{\cdotv})}\prod_{k=1}^{K} \left\{r_k^{\ell_k}\sum_{(n_{k1},\ldots,n_{k\ell_k})\in\mathcal{D}_{n_k,\ell_k}} \prod_{t=1}^{\ell_k} \frac{1}{n_{kt}} \right\}. \notag
\end{align}
Summing over all $\nv$ in the set $\mathcal{M}_{n,K}=\left\{(n_1,\ldots,n_K): n_k\ge 0\text{ and }\sum_{k=1}^K n_k=n\right\}$, we have
\begin{align}
P(\ellv\,|\,n,\rv) &= \sum_{(n_1,\ldots,n_K)\in\mathcal{M}_{n,K}} \frac{n!}{\prod_{k=1}^K \ell_k!}\frac{\Gamma(r_{\cdotv})}{\Gamma(n+r_{\cdotv})}\prod_{k=1}^{K} \left\{ r_k^{\ell_k}\sum_{(n_{k1},\ldots,n_{k\ell_k})\in\mathcal{D}_{n_k,\ell_k}} \prod_{t=1}^{\ell_k} \frac{1}{n_{kt}} \right\}\notag\\
& = \left\{\frac{n!}{\prod_{k=1}^K \ell_k!}\frac{\Gamma(r_{\cdotv})}{\Gamma(n+r_{\cdotv})}\prod_{k=1}^{K} r_k^{\ell_k} \right\} \left\{\sum_{(n_1,\ldots,n_K)\in\mathcal{M}_{n,K}}\sum_{(n_{k1},\ldots,n_{k\ell_k})\in\mathcal{D}_{n_k,\ell_k}} \prod_{t=1}^{\ell_k} \frac{1}{n_{kt}}\right\} \notag\\
&= \left\{\frac{n!}{\prod_{k=1}^K \ell_k!}\frac{\Gamma(r_{\cdotv})}{\Gamma(n+r_{\cdotv})}\prod_{k=1}^{K} r_k^{\ell_k}\right\} 
 \left\{\sum_{(n_{1},\ldots,n_{\ell_{\cdotv}})\in\mathcal{D}_{n,\ell_{\cdotv}}} \prod_{t=1}^{\ell_{\cdotv}} \frac{1}{n_{t}} \right\}\notag\\
& = \frac{\ell_{\cdotv}!}{\prod_{k=1}^K \ell_k!}\frac{\Gamma(r_{\cdotv})}{\Gamma(n+r_{\cdotv})}|s(n,\ell_{\cdotv})|\prod_{k=1}^{K} r_k^{\ell_k} \notag\\
& = \mbox{Mult}(\ellv;\ell_{\cdotv},r_1/r_{\cdotv},\ldots,r_K/r_{\cdotv}) \mbox{CRT}(\ell_{\cdotv};n,r_{\cdotv})\,\, . \notag
\end{align}

\end{proof}

\section{Posterior analysis for the hGNBP} 
\label{app:post}
 Denote $\ell\sim\mbox{CRT}(n,r)$ as the Chinese restaurant table (CRT) random variable 
generated as the summation of $n$ independent Bernoulli random variables as
$
\ell=\sum_{i=1}^{n} h_i,~h_i\sim\mbox{Bernoulli}\left(\frac{r}{r+i-1}\right).\notag
$
 The probability mass function (PMF) of 
 $\ell\sim\mbox{CRT}(n,r)$ can be expressed as
$
 f_L(\ell\,|\,n,r) = \frac{\Gamma(r)r^\ell}{\Gamma(n+r)}|s(n,\ell)|,\notag
$
where $\ell\in\{0,1,\ldots,n\}$ and $$|s(n,\ell)| = \frac{n!}{\ell!}\sum_{(n_1,\ldots,n_{\ell})\in \mathcal{D}_{n,\ell}}\prod_{t=1}^\ell\frac{1}{n_t}$$
are unsigned Stirling numbers of the first kind \citep{johnson1997discrete} that can be obtained by summing over the elements of the set $\mathcal{D}_{n,\ell}=\{(n_1,\ldots,n_{\ell}): n_t\ge 1\text{ and }\sum_{t=1}^\ell n_t = n\}$. 

Let $u\sim\mbox{Logarithmic}(p)$ denote the logarithmic distribution \citep{Fisher1943} with PMF
$
f_U(u\,|\,p) = \frac{1}{-\ln(1-p)}\frac{p^u}{u},\notag
$
where $u\in\{1,2,\ldots\}$.
Denote $x \sim\mbox{SumLog}(\ell,p)$ as the sum-logarithmic random variable \citep{NBP_CountMatrix} generated as
$x=\sum_{t=1}^{\ell} u_t,~u_t\sim\mbox{Logarithmic}(p)$, with PMF $
f_{N}(n \,|\, \ell,p) = \frac{ p^n \ell! \ |s(n,\ell)|}{n! \ [-\ln(1-p)]^\ell} .
$
As revealed in \citet{NBP2012},
the joint distribution of $n$ and $\ell$ given $r$ and $p$ in 
$
\ell\,|\, n \sim\mbox{CRT}(n,r),~n\sim\mbox{NB}(r,p)\notag
$
 is the same as that in
$
n \,|\, \ell \sim \mbox{SumLog}(\ell,p), 
 ~\ell\sim\mbox{Pois}[-r\ln(1-p)],
$
which is called as the Poisson-logarithmic bivariate distribution, with PMF 
\beq
f_{N,L}(n,\ell\,|\,r,p)=\frac{|s(n,\ell)|r^\ell}{n!}p^n(1-p)^r.\label{eq:Po-log}
\eeq

Denote $L\sim\mbox{CRTP}(X,G)$ as a CRT process such that
$ 
 L(A)=\sum_{\omega\in A}L(\omega), L(\omega)\sim\mbox{CRT}[X(\omega),G(\omega)]
$
for each $A\subset\Omega$, and $X\sim\mbox{SumLogP}(L,p)$ as a sum-logarithmic process such that
$
X(A) \sim \mbox{SumLog}[L(A),p]
$
for each $A\subset\Omega$. 
 As in \citet{NBP2012}, generalizing the Poisson-logarithmic bivariate distribution, 
 one may show that $X$ and $L$ given $G$ and $p$ in 
$$L \,|\,X \sim\mbox{CRTP}(X,G),~X\sim\mbox{NBP}(G,p)$$
is equivalent in distribution to those in 
$$X \,|\, L \sim\mbox{SumLogP}(L,p),~L\sim\mbox{PP}[-G\ln(1-p)],$$
where $L\sim\mbox{PP}[-G\ln(1-p)]$ is a Poisson process such that $L(A)\sim\mbox{Pois}[-G(A)\ln(1-p)]$ for each $A\subset\Omega$.
Generalizing the analysis for the GNBP in \citet{NBP2012} and \citet{NBP_CountMatrix}, with 
$$\tilde{p}_j := \frac{-\ln(1-p_j)}{c_j-\ln(1-p_j)},~~\tilde{\tilde{p}}: = \frac{-\sum_j \ln(1-\tilde{p}_j)}{c_0-\sum_j\ln(1-\tilde{p}_j)},$$
we can express the conditional posteriors of $G$ and $\Theta_j$ as
\beqs
&(L_j\,|\,X_j,\Theta_j)\sim\mbox{CRTP}(X_j,\Theta_j),~~(\widetilde{L}_j\,|\,L_j,G)\sim\mbox{CRTP}(L_j,G),\notag\\
&(G\,|\,\{\widetilde{L}_j,p_j\}_j,G_0)\sim\Gamma\mbox{P}\left(G_0+\sum_j \widetilde{L}_j,[{c_0-\sum_j\ln(1-\tilde{p}_j)}]^{-1}\right),\notag\\
&(\Theta_j\,|\,G,\widetilde{L}_j,p_j,c_j)\sim\Gamma\mbox{P}\left(G+ {L}_j,[{c_j-\ln(1-p_j)}]^{-1}\right).\label{eq:post_analysis}
\eeqs 
If we let $\gamma_0\sim\mbox{Gamma}(a_0,1/b_0)$, the conditional posterior of $\gamma_0$ can be expressed as
\beqs
&(\widetilde{\widetilde{L}}\,|\,\{\widetilde{L}_j\}_j,G_0)\sim\mbox{CRTP}(\sum_j \widetilde{L}_j,G_0),\notag\\
&(\gamma_0\,|\,\widetilde{\widetilde{L}}, \{p_j,c_j\}_j,c_0)\sim\mbox{Gamma}\left(a_0+\widetilde{\widetilde{L}}(\Omega),[{b_0-\ln(1-\tilde{\tilde{p}})}]^{-1}\right). \notag
\eeqs
 If the base measure $G_0$ is finite and continuous, we have $$\widetilde{\widetilde{L}}(\Omega)=\sum_{k=1}^\infty \delta\left(\sum_j\widetilde{L}_{j}(\phiv_k)>0\right)=\sum_{k=1}^\infty \delta\left(\sum_j X_{j}(\phiv_k)>0\right),$$ which is the number of active atoms that are associated with nonzero counts, otherwise we have $\big(\widetilde{\widetilde{L}}(\omega_k)\,|\,\{\widetilde{L}_j\}_j,G_0\big)\sim\mbox{CRT}\big[\sum_j \widetilde{L}_j (\omega_k),G_0(\omega_k)\big]$ for all atoms $\omega_k\in \Omega$. 
 In this paper, we let $K^+= \widetilde{\widetilde{L}}(\Omega)$ denote the number of active atoms.

\section{Additional Gibbs sampling update equations}

\begin{algorithm}[h]
 \caption{Gibbs sampling algorithms for the hierarchical gamma-negative binomial process negative binomial factor analysis (
 Dirichlet-multinomial mixed-membership model).
 }\label{tab:algorithm}
 \begin{algorithmic}[1] 
 \For{\text{$iter=1: MaxIter$ 
 }} Gibbs sampling
\algnewcommand\algorithmicswitch{\textbf{switch}}
\algnewcommand\algorithmiccase{\textbf{case}}
\algnewcommand\algorithmicassert{\texttt{assert}}
\algnewcommand\Assert[1]{\State \algorithmicassert(#1)}%
\algdef{SE}[SWITCH]{Switch}{EndSwitch}[1]{\algorithmicswitch\ #1\ \algorithmicdo}{\algorithmicend\ \algorithmicswitch}%
\algdef{SE}[CASE]{Case}{EndCase}[1]{\algorithmiccase\ #1}{\algorithmicend\ \algorithmiccase}%
\algtext*{EndSwitch}%
\algtext*{EndCase}%
\Switch{$Gibbs~sampler$}
 \Case{$regular~blocked~Gibbs~sampler$}
 \State \text{sample $\{z_{ji}\}_{j,i}$ and then calculate $\{n_{vjk}\}_{v,j,k}$}; \text{sample a latent count $\ell_{vjk}$ for each $n_{vjk}$};
 \EndCase
 \Case{$collapsed~Gibbs~sampler$}
 \State \text{sample $\{z_{ji},b_{ji}\}_{j,i}$ and then calculate $\{n_{vjk},\ell_{vjk}\}_{v,j,k}$};
 \EndCase
 \Case{$compound~Poisson~based~blocked~Gibbs~sampler$}
 \State \text{sample a latent count $\ell_{vj}$ for each $n_{vj}$}; \text{sample $\{\ell_{vjk}\}_{k}$ for each $\ell_{vj}$};
 \EndCase
 \EndSwitch
 \State \textbf{end switch}
 \State \text{sample $\{p_j\}_{j}$}; \text{sample $\{c_j\}_{j}$}; \text{sample $\gamma_0$; sample $c_0$}; relabel the 
 active factors from 1 to $K^+$. 
\If {$collapsed~Gibbs~sampler$}
 \State \text{sample $\{r_k\}_{1,K^+}$; sample $r_{\star}$; sample $\{\theta_{\cdotv j}\}_j$};
 \Else
\For{\texttt{$k=1:K^++K_\star$}}
 \State 
 sample $\phiv_k^{(t)}$; sample $r_k$; sample $\{\theta_{kj}\}_j$;
 \EndFor

\EndIf
 
 \EndFor
 \end{algorithmic}
 \normalsize
\end{algorithm}
%

\subsection{Additional update equations for blocked Gibbs sampling}\label{app:hGNBP_1}

\emph{\textbf{Sample $p_j$.}} We sample $p_j$ as
\beq
(p_j\,|\,-)\sim\mbox{Beta}(a_0 + n_j,b_0+\theta_{\cdotv j}).\label{eq:sampletpj}
\eeq
\emph{\textbf{Sample $c_j$.}} We sample $c_j$ as
\beq
(c_j\,|\,-)\sim \mbox{Gamma}\left[e_0 + G(\Omega),1/(f_0+\theta_{\cdotv j})\right]. \notag
\eeq
\emph{\textbf{Sample $r_k$}}. 
We first sample latent counts 
and then sample $\gamma_0$ and $c_0$ as
\begin{align}
&(\tilde{\ell}_{jk}\,|\,-)\sim\mbox{CRT}(\ell_{\cdotv jk},r_k),~(\gamma_0\,|\,-)\sim \mbox{Gamma}\left(a_0+K^+,\frac{1}{b_0-\ln(1-\tilde{\tilde{p}})}\right),\notag\\
&(c_0\,|\,-)\sim \mbox{Gamma}\left(e_0 + \gamma_0,\frac{1}{f_0+G(\Omega)}\right), \notag 
\end{align}
%
where $\tilde{\ell}_{\cdotv k} := \sum_j \tilde{\ell}_{j k}$, $K^+:=\sum_{k}\delta(n_{\cdotv \cdotv k}>0) = \sum_{k}\delta(\tilde{\ell}_{\cdotv k} >0) $. 
For all the points of discontinuity, $i.e.$, the factors in the set $\{\phiv_k\}_{k:n_{\cdotv \cdotv k}>0}$, we relabel their indices from 1 to $K^+$ and then sample $r_k$ as
\beq
(r_k\,|\,-)\sim \mbox{Gamma}\left(\tilde{\ell}_{\cdotv k},\frac{1}{c_0-\sum_j\ln(1-\tilde{p}_j)}\right), \notag
\eeq
and for the absolutely continuous space $\{\phiv_k\}_{k:n_{\cdotv \cdotv k}=0}$, we draw $K_\star$ unused atoms, whose weights are sampled as 
\beq
(r_k\,|\,-)\sim \mbox{Gamma}\left(\frac{\gamma_0}{K_\star} ,\frac{1}{c_0-\sum_j\ln(1-\tilde{p}_j)}\right). \notag 
\eeq
We let $K:=K^+ +K_\star$ and $G(\Omega):=\sum_{k=1}^{K} r_k $.\\ 
\emph{\textbf{Sample $\phiv_{k}$.}} Denote $\ell_{v\cdotv k}=\sum_{j=1}^J \ell_{vj k}$. Since $\ell_{vjk}\sim\mbox{Pois}[-\phi_{vk}\theta_{kj}\ln(1-p_j)]$ in the prior, we can sample $\phiv_{k}$ as 
\beq
(\phiv_k\,|\,-)\sim\mbox{Dir}(\eta+\ell_{1\cdotv k},\ldots,\eta+\ell_{V\cdotv k}). \notag 
\eeq
\emph{\textbf{Sample $\theta_{kj}$.}} Denote $\ell_{\cdotv j k}=\sum_{v=1}^V \ell_{vj k}$. We can sample $\theta_{kj}$ as 
\beq
(\theta_{kj}\,|\,-)\sim\mbox{Gamma}\left(r_k + \ell_{\cdotv j k},\frac{1}{c_j-\ln(1-p_j)}\right).\label{eq:sampletheta}
\eeq

\subsection{Additional update equations for collapsed Gibbs sampling}\label{app:hGNBP_2}

The other model parameters can all be sampled in the way similar to how they are sampled in Section \ref{sec:blockGibbs}. Below we highlight the differences. First, we do not need to sample $\{\phiv_k\}$. Instead of sampling $\{\theta_{kj}\}_k$, 
%
%
 we only need to sample $\theta_{\cdotv j}$ as 
\beq
(\theta_{\cdotv j}\,|\,-)\sim\mbox{Gamma}\left[G(\Omega) +\textstyle\sum_k \ell_{\cdotv j k},1/(c_j-\ln(1-p_j)\right]. \notag
\eeq 
For the absolutely continuous space, we have 
\beq 
(r_\star\,|\,-)\sim \mbox{Gamma}\left(\gamma_0,\frac{1}{c_0-\sum_j\ln(1-\tilde{p}_j)}\right). \notag
\eeq
We have $K^+=\sum_k \delta(\ell_{\cdotv \cdotv k}>0)$ and $G(\Omega) := r_\star + \sum_{k:\,\ell_{\cdotv \cdotv k}>0} r_k $.

\section{Gamma-negative binomial process PFA and DCMLDA}\label{app:GNBP}
We consider the GNBP \citep{NBP2012} as
\vspace{-3mm}\beq
X_j\sim\mbox{NBP}(G,p_j),~G\sim\Gamma\mbox{P}(G_0,1/c_0).\label{eq:GNBP}\vspace{-3mm}
\eeq
The GNBP multinomial mixed-membership model of \citet{NBP2012} can be expressed as
\vspace{-3mm}\begin{align}
&x_{ji}\sim\mbox{Cat}(\phiv_{z_{ji}}),~z_{ji}\sim\mbox{Cat}(\thetav_j/\theta_{\cdotv j}),\notag\\
&\theta_{kj}\sim\mbox{Gamma}\left[r_k,{p_j}/{(1-p_j)}\right],\notag\\
&n_j\sim\mbox{Pois}(\theta_{\cdotv j}), ~p_j\sim\mbox{Beta}(a_0,b_0),\notag\\
&G\sim\Gamma\mbox{P}(G_0,1/c_0), \notag 
\vspace{-3mm}\end{align}
which, as far as the conditional posteriors of $\{\phiv_k\}_k$ and $\{\thetav_j\}_j$ are concerned, can be equivalently represented as the GNBP-PFA 
\begin{align}
&n_{vj}=\sum_{k=1}^\infty n_{vjk}, ~~n_{vjk}\sim\mbox{Pois}(\phi_{vk}\theta_{kj}).\notag\\
&\theta_{kj}\sim\mbox{Gamma}\left[r_k,{p_j}/{(1-p_j)}\right],\notag\\
&p_j\sim\mbox{Beta}(a_0,b_0),~G\sim\Gamma\mbox{P}(G_0,1/c_0). \notag 
\end{align}
Similar to how adaptive truncation is used in blocked Gibbs sampling for the hGNBP-NBFA, one may readily extend the blocked Gibbs sampler for the GNBP multinomial mixed-membership model developed in \citet{NBP2012}, which has a fixed finite truncation, to a one with adaptive truncation. We omit these details for brevity. We 
describe a collapsed Gibbs sampler for the GNBP-PFA in Appendix \ref{app:GNBP-PFA}.

As discussed before, the GNBP can also be applied to DCMLDA to support countably infinite factors. 
We express the GNBP-DCMLDA as
\vspace{-2mm}\begin{align}
&x_{ji}\sim\mbox{Cat}(\phiv^{[j]}_{z_{ji}}),~z_{ji}\sim\mbox{Cat}(\thetav_j),\notag\\
&\phiv^{[j]}_k\sim\mbox{Dir}(\phiv_k r_k),~\thetav_j\sim\mbox{Dir}(\rv),\notag\\
&n_j\sim\mbox{NB}(r_{\cdotv},p_j), ~p_j\sim\mbox{Beta}(a_0,b_0),\notag\\
&G\sim\Gamma\mbox{P}(G_0,1/c_0), \label{eq:GNBP DCMLDA}
\vspace{-3mm}
\end{align}
which, as far as the conditional posteriors of $\{\phiv_k\}_k$ and $\{r_k\}_k$ are concerned, can be equivalently represented as
\vspace{-3mm}\begin{align}
&n_{vj}=\sum_{k=1}^\infty n_{vjk}, ~~n_{vjk}\sim\mbox{NB}(\phi_{vk}r_k,\, p_j).\notag\\
&p_j\sim\mbox{Beta}(a_0,b_0),~G\sim\Gamma\mbox{P}(G_0,1/c_0).\label{eq:NBFA1_1}\vspace{-3mm}
\end{align}
The restriction is evident from (\ref{eq:NBFA1_1}) as all the samples are enforced to have the same factor scores as $\rv_k$ under the shared factors $\{\phiv_k\}_k$. Blocked Gibbs sampling with and without sampling $z_{ji}$ 
for the GNBP-DCMLDA
can be similarly derived as those for the hGNBP DMMM model, omitted here for brevity. 
 We describe in detail a collapsed Gibbs sampler
 for the GNBP-DCMLDA in Appendix \ref{app:DCMLDA}.

\subsection{Collapsed Gibbs sampling for GNBP-PFA}\label{app:GNBP-PFA}

For the GNBP in \eqref{eq:GNBP}, the conditional likelihood $p(\{X_j\}_{1,J}\,|\, G)$ is shown in Appendix B.1 of \citet{NBP_CountMatrix}. As there is a one-to-many mapping from $\{X_j\}_{1,J}$ to $\zv=\{z_{11},\ldots,z_{Jm_J}\}$, similar to the analysis in \citet{BNBP_EPPF}, we have the joint likelihood of $\zv$ and the sample lengths $\mv=(m_1,\ldots,m_J)$
as 
 \begin{align}\label{eq:Likelihood}
p(\zv,\mv\,|\, G,\pv) = \frac{p(\{X_j\}_{1,J}\,|\, G)}{ \prod_{j=1}^J \frac{n_j!}{{\prod_{k=1}^{\infty}} n_{jk}!}} 
&= 
\prod_{j=1}^J\frac{p_j^{n_{j}}}{n_j!}\prod_{k=1}^{\infty} \frac{\Gamma(n_{jk}+r_k)}{\Gamma(r_k)}(1-p_j)^{r_k}. \notag 
\end{align}
Assuming the $K^+$ factors that are associated with nonzero counts are relabeled in an arbitrary oder from $1$ to $K^+$, based on this conditional likelihood, we have a prediction rule conditioning on $G$ as
\begin{align}
P(z_{ji}=k\,|\, \zv^{-ji},\mv,G) &\propto
\begin{cases} \vspace{0.15cm} n_{jk}^{-ji}+r_k
, & \mbox{for } k=1,\cdots,(K^+)^{-ji} ; \\
r_{\star}, 
& \mbox{if } k=(K^+)^{-ji}+1,
\end{cases} \notag 
\end{align}
where $r_{\star}=G(\Omega\backslash\{\phiv_k\}_{1,K^+})$ is the total weight of all the factors assigned with zero count. 
This prediction rule becomes very similar to the direct assignment sampler of the hierarchical Dirichlet process \citep{HDP} if one writes each $r_k$ as the product of a total random mass $\alpha$ and a probability $\pi_k$, with $\alpha = \sum_{k=1}^\infty r_k$ and $\sum_{k=1}^\infty \pi_k=1$. This is as expected since the gamma process can be represented as the independent product of a gamma process and a Dirichlet process, under the condition that the mass parameter of the gamma process is the same as the concentration parameter of the Dirichlet process, and the GNBP is closely related to the hierarchical Dirichlet process for mixed-membership modeling \citep{NBP2012}. 

Similar to the derivation of collapsed Gibbs sampling for the mixed-membership model based on the beta-negative binomial process, as shown in \cite{BNBP_EPPF}, we can write the collapsed Gibbs sampling update equation for the factor indices as
\begin{align} 
P(z_{ji}=k|\xv, \zv^{-ji},\mv,G) &\propto
\begin{cases}\displaystyle \vspace{0.15cm} \frac{\eta+n_{v_{ji}\cdotv k}^{-ji}}{V\eta+ n_{\cdotv k}^{-ji}}\cdotv (n_{jk}^{-ji}+r_k)
, & \mbox{for } k=1,\ldots,(K^+)^{-ji} ; \\
\displaystyle\frac{1}{ V}\cdotv r_{\star}, 
& \mbox{if } k=(K^+)^{-ji}+1;
\end{cases} \notag
\end{align}
and if $k= (K^+)^{-ji}+1$ happens, 
we draw $\beta\sim\mbox{Beta}(1,\gamma_0)$ and then let $r_k =\beta r_{\star} $ and $r_{\star} =(1-\beta) r_{\star}$. 
Gibbs sampling update equations for the other model parameters of the GNBP can be similarly derived as in \citet{NBP2012} and \citet{NBP_CountMatrix}, omitted here for brevity. 

\subsection{Collapsed Gibbs sampling for GNBP-DCMLDA}\label{app:DCMLDA}
For collapsed Gibbs sampling of (\ref{eq:GNBP DCMLDA}), introducing the auxiliary variables
\beq
\{b_{ji}\}_{i:x_{ji}=v,z_{ji}=k}\sim\mbox{CRP}(n_{vjk},\phi_{vk}r_k), \notag 
\eeq
 we have
the joint likelihood of $\bv_j$, $\zv_j$, $\xv_j$ and $n_j$ for DCMLDA as
\beq
P(\bv_j,\xv_j,\zv_j,n_j\,|\,\Phimat,\thetav_j,p_j)=\frac{1}{n_j!}\prod_v\prod_{k} \left\{{(\phi_{vk}r_k)^{\ell_{vjk}}}{ p_j^{n_{vjk}}(1-p_j)^{\phi_{vk}r_k}} \prod_{t=1}^{\ell_{vjk}}\Gamma(n_{vjkt})\right\}, \notag 
\eeq
where $\ell_{vjk} $ is the number of unique indices in $\{b_{ji}\}_{i:x_{ji}=v,z_{ji}=k}$ and $n_{vjkt} =\sum_{i=1}^{n_j} \delta(x_{ji}=v,z_{ji}=k,b_{ji}=t)$.

Marginalizing out $\Phimat$ from this likelihood, we have
\begin{align}
P(\{\bv_j,\xv_j,\zv_j,n_j\}_j\,|\,G,\pv) = 
&\, e^{ r_{\star} \sum_j \ln(1-p_j)}\left\{\prod_j p_j^{n_j}\frac{\prod_v\prod_k\left(\prod_{t=1}^{\ell_{vjk}}\Gamma(n_{vjkt})\right)}{n_j!}\right\}\notag\\
&\times\left\{\prod_{k:\,\ell_{\cdotv \cdotv k}>0} r_k^{\ell_{\cdotv \cdotv k}} e^{r_k\sum_j \ln(1-p_j)} \frac{\Gamma(V\eta)}{\Gamma(\ell_{\cdotv\cdotv k}+V\eta)}\prod_{v=1}^V\frac{\Gamma(\ell_{v\cdotv k}+\eta)}{\Gamma(\eta)} \right\},\label{eq:Like_1}
\end{align}
where $r_\star :=\sum_{k:\,\ell_{\cdotv \cdotv k}=0} r_k$.
With this likelihood, we have
\begin{align} 
P(z_{ji}=k,b_{ji}=t\,|\,x_{ji},\zv^{-ji},\bv^{-ji},G) &\propto
\begin{cases}\vspace{.15cm}
\displaystyle n_{x_{ji}jkt}^{-ji}
, & \emph{\mbox{if }} k\le (K^+)^{-ji}, t\le \ell_{x_{ji}jk}^{-ji} ; \\ 
\displaystyle \frac{\ell_{x_{ji}\cdotv k}^{-ji}+\eta}{\ell_{\cdotv\cdotv k}^{-ji}+V\eta}r_k, & \emph{\mbox{if }} k\le (K^+)^{-ji}, t= \ell_{x_{ji}jk}^{-ji}+1 ;\\
\vspace{.15cm}
\displaystyle \frac{r_{\star}}{V}, & \emph{\mbox{if }} k= (K^+)^{-ji}+1, t=1 ; \end{cases} \notag
\end{align}
and if $k= (K^+)^{-ji}+1$ happens, then we draw then we draw $\beta\sim\mbox{Beta}(1,\gamma_0)$ and then let $r_k =\beta r_{\star} $ and $r_{\star} =(1-\beta) r_{\star}$. 

Using the Palm formula \citep{james2002poisson,james2009posterior, CarTehMur2013a}, similar to related derivation in \citet{NBP_CountMatrix}, we may further marginalize out $G$ from (\ref{eq:Like_1}), leading to
\begin{align}
&P(\{\bv_j,\xv_j,\zv_j,n_j\}_j\,|\,\gamma_0,c_0,\pv) \notag\\
&=\gamma_0^{K^+} e^{ -\ln\left(\frac{c_0- \sum_j \ln(1-p_j)}{c_0}\right)}\left\{\prod_j p_j^{n_j}\frac{\prod_v\prod_k\left(\prod_{t=1}^{\ell_{vjk}}\Gamma(n_{vjkt})\right)}{n_j!}\right\}\notag\\
&\,\,\,\,\,\times\left\{\prod_{k:\,\ell_{\cdotv \cdotv k}>0} \frac{\Gamma(\ell_{\cdotv \cdotv k})}{[c_0-\sum_j\ln(1-p_j)]^{\ell_{\cdotv \cdotv k}}} \frac{\Gamma(V\eta)}{\Gamma(\ell_{\cdotv\cdotv k}+V\eta)}\prod_{v=1}^V\frac{\Gamma(\ell_{v\cdotv k}+\eta)}{\Gamma(\eta)} \right\}
, \notag 
\end{align}
with which we have
\begin{align} 
&P(z_{ji}=k,b_{ji}=t\,|\,x_{ji},\zv^{-ji},\bv^{-ji},\gamma_0,c_0)\notag\\
 &\propto
\begin{cases}\vspace{.15cm}
\displaystyle n_{x_{ji}jkt}^{-ji}
, & \emph{\mbox{if }} k\le (K^+)^{-ji}, t\le \ell_{x_{ji}jk}^{-ji} ; \\ 
\displaystyle \frac{\ell_{x_{ji}\cdotv k}^{-ji}+\eta}{\ell_{\cdotv\cdotv k}^{-ji}+V\eta} \frac{ \ell_{\cdotv \cdotv k}^{-ji}}{c_0-\sum_j\ln(1-p_j)}, & \emph{\mbox{if }} k\le (K^+)^{-ji}, t= \ell_{x_{ji}jk}^{-ji}+1 ;\\
\vspace{.15cm}
\displaystyle \frac{1}{V} \frac{ \gamma_0}{c_0-\sum_j\ln(1-p_j)}, & \emph{\mbox{if }} k= (K^+)^{-ji}+1, t=1. \end{cases} \notag
\end{align}
We use the above equation in the collapsed Gibbs sampler for GNBP-DCMLDA.

\section{Sample the Dirichlet smoothing parameter }\label{app:eta}

 For the hGNBP-NBFA, from (\ref{eq:fullyaugmented}), we have the likelihood for $\{\phiv_k\}$ as
 \beq
 \mathcal{L}(\{\phiv_k\})\propto \prod_k \mbox{Mult}( \ell_{1\cdotv k},\ldots, \ell_{V\cdotv k} ; \ell_{\cdotv\cdotv k}, \phiv_k) \label{eq:philike}
 \eeq
 Marginalizing out $\phiv_k$ from (\ref{eq:philike}), we have the likelihood for $\eta$ as
 \beq
 \mathcal{L}(\eta)\propto\prod_k\mbox{DirMult}( \ell_{1\cdotv k},\ldots, \ell_{V\cdotv k};\ell_{\cdotv\cdotv k},\eta,\ldots,\eta). \notag
 \eeq
Since the product of $ \mathcal{L}(\eta)$ and $\prod_k \mbox{Beta}(q_k;\ell_{\cdotv\cdotv k},\eta V)$ can be expressed as
 \beq
\mathcal{L}(\eta)\prod_k \mbox{Beta}(q_k;\ell_{\cdotv\cdotv k},\eta V) \propto
\prod_k \prod_v \mbox{NB}(\ell_{v\cdotv k}; \eta,q_k), \notag
 \eeq
 we can further apply the data augmentation technique for the NB distribution of \citet{NBP2012} to derive closed-form update equations for $\eta$ as
\begin{align}
 &(q_k\,|\,-)\sim\mbox{Beta}( \ell_{\cdotv \cdotv k}, V\eta),~~~(t_{vk}\,|\,-)\sim\mbox{CRT}(\ell_{v \cdotv k},\eta),\notag\\
 &(\eta\,|\,-)\sim\mbox{Gamma}\left(a_0 + \sum_{v=1}^V\sum_{k=1}^{K^+} t_{vk},\frac{1}{b_0-V\sum_{k=1}^{K^+}\ln(1-q_k)}\right) \label{eq:sampleeta}
 \end{align}

To sample $\eta$ for the GNBP-PFA, we simply replace $\ell_{\cdotv \cdotv k}$ and $\ell_{v \cdotv k}$ in (\ref{eq:sampleeta}) with $n_{\cdotv \cdotv k}$ and $n_{v \cdotv k}$, respectively. We note the inference of $\eta$ for the GNBP-PFA can be related to the inference of that 
 for LDA 
described in \citet{newman2009distributed}.

 \section{Comparisons of different sampling strategies}\label{app:sampler}
We first diagnose the convergence of the regular blocked Gibbs sampler in Section \ref{sec:blockGibbs}, the collapsed Gibbs sampler in Section \ref{sec:collapseGibbs}, and the compound Poisson representation based blocked Gibbs sampler in Section \ref{sec:blockGibbs_1} for the hGNBP-NBFA (Dirichlet-multinomial mixed-membership model), 
via the trace plots of the inferred number of active factors $K^+$.
We set the Dirichlet smoothing parameter as $\eta=0.05$, and initialize the number of factors as $K=0$ for the collapsed Gibbs sampler and $K=10$ for both the regular and compound Poisson based blocked Gibbs samplers. We also consider initializing the number of factors as $K=500$ for all three samplers.

\begin{figure}[!b]
\begin{center}
\includegraphics[width=0.8\columnwidth]{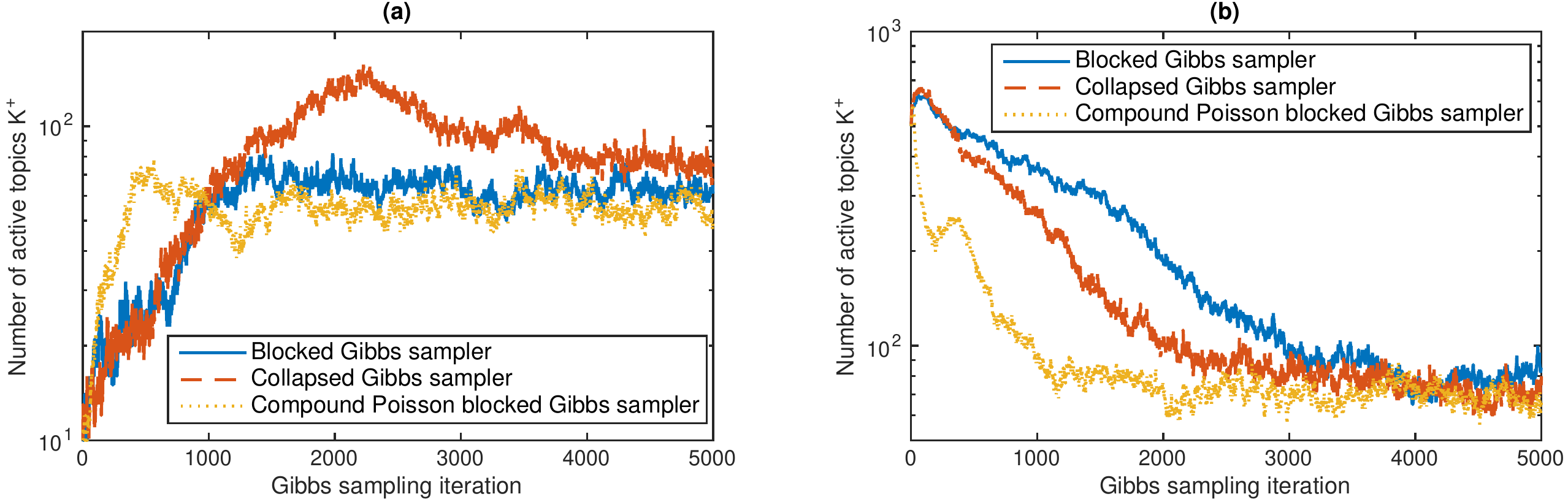}
\end{center}
\vspace{-6mm}
\caption{ \label{fig:PsyReview_Trace} \small
Comparison of three different Gibbs samplers for the hierarchical gamma-negative binomial process negative binomial factor analysis (hGNBP-NBFA) on the PsyReview dataset, with the number of factors initialized as (a) $K=0$ for the collapsed sampler and $K=10$ for both blocked samplers, and (b) $K=500$ for all three samplers. 
 In each plot, the blue, red, and yellow curves correspond to the active number of factors as a function of Gibbs sampling iteration for 
 the regular blocked Gibbs sampler, the collapsed Gibbs sampler, and the compound Poisson representation based blocked Gibbs sampler, respectively. 
}\vspace{-0.5mm}
\end{figure}
%
%

As shown in Figure \ref{fig:PsyReview_Trace} for the PsyReview dataset, both the regular blocked Gibbs sampler and collapsed Gibbs sampler travel relatively slowly to the target distribution of the number of active factors $K^+$, especially when the number of factors is initialized to be large, 
whereas the compound Poisson based blocked Gibbs sampler travels relatively quickly to the target distribution 
in both cases. 
We have also made similar comparisons on 
both the JACM and NIPS12 datasets, and the experiments on all three datasets consistently suggest that the compound Poisson representation based blocked Gibbs sampler converges the fastest in the number of inferred active factors. 



%
%

We observe similar differences in convergence between the blocked Gibbs sampler, the collapsed Gibbs sampler, and the compound Poisson representation based blocked Gibbs sampler for the GNBP-DCMLDA. This is as expected since GNBP-DCMLDA can be considered as a special case of the hGNBP-NBFA, and its compound Poisson representation also allows it to eliminate the need of sampling the factor indices $\{z_{ji}\}$.


For the GNBP-PFA, we find that its blocked Gibbs sampler, presented in \citet{NBP2012} and improved in this paper to allow adaptively truncating the number of active factors in each Gibbs sampling iteration, could converge slightly faster if the number of factors is initialized to be large. However, its collapsed Gibbs sampler shown in the Appendix often converges much faster in the number of inferred active factors 
if the number of factors is initialized with a small value. 

Therefore, to learn the factors in all the following experiments, we use the compound Poisson representation based blocked Gibbs sampler for both the hGNBP-NBFA and GNBP-NBFA, and use collapsed Gibbs sampling for the GNBP-PFA.


\section{Additional table and plots}\label{app:plot}

\begin{table}[h]
\begin{scriptsize}
\caption{Datasets used in experiments.}\label{tab:data}
\vspace{-4mm}
\begin{center}
\begin{tabular}{l | r r r r r r}
\toprule
 &JACM &PsyReview & NIPS12 & $atheism$ vs & $pc$ vs  & 20 news \\
  & & &  & $religion$ & $mac$ & groups \\
\midrule
Number of unique covariates $V$& 1,539 & 2,566 & 13,649 & 6,509 &4,737 & 33,420 \\
Number of samples& 536 & 1,281 & 1,740 & 856 &1,162 & 11,269 \\
Total number of covariate {indices} & 68,055 & 71,279 & 2,301,375 & 115,904 & 90,667 &1,424,713 
 \\
Average sample length & 127 &56 & 1,323& 135 & 78 &126\\
\bottomrule
\end{tabular}
\end{center}
\end{scriptsize}
\end{table}

We show the results of the GNBP-NBFA, GNBP-DCMLDA, and hGNBP-NBFA on both the PsyReview and JACM datasets in the following figures. 

\begin{figure}[!h]
\begin{center}
\includegraphics[width=1\columnwidth]{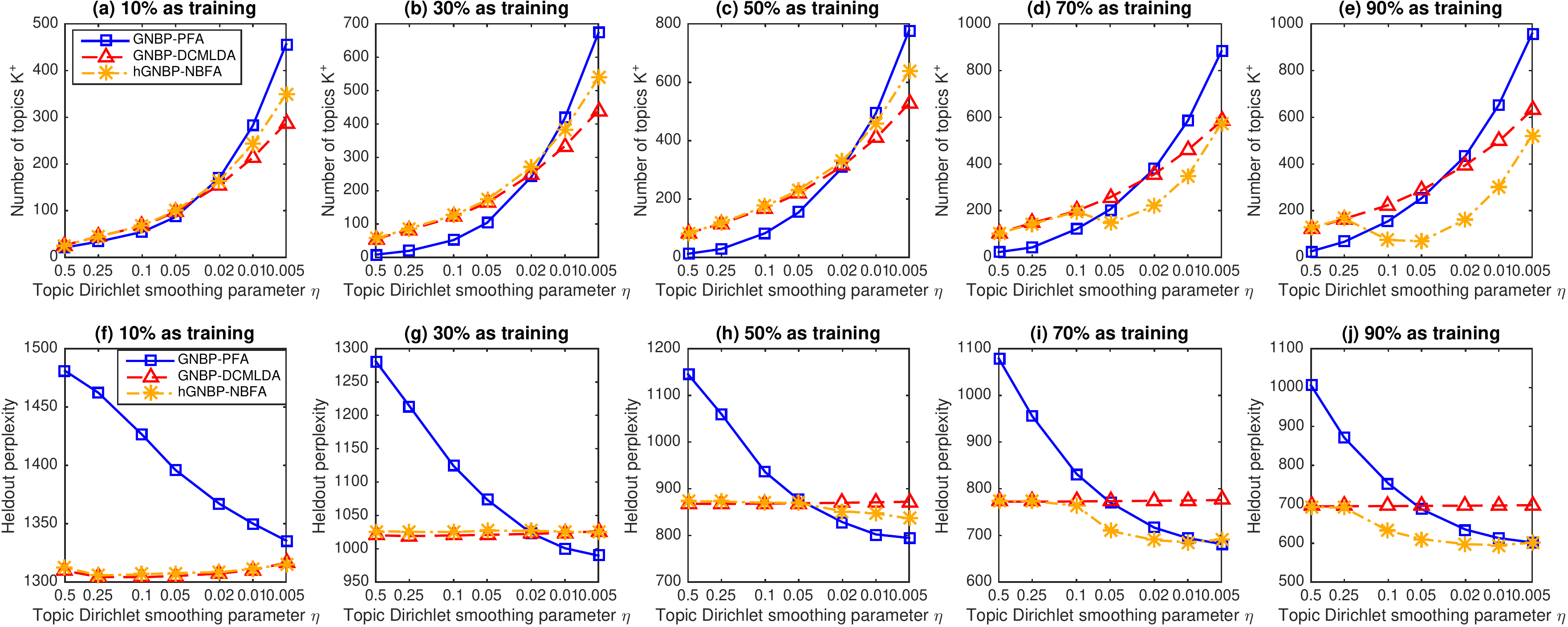}
\end{center}
\vspace{-6.0mm}
\caption{\small \label{fig:perplexity3_1}
Analogous plots to those in Figure \ref{fig:perplexity1_1} for the PsyReview dataset. 
}
\vspace{5mm}
\begin{center}
\includegraphics[width=1\columnwidth]{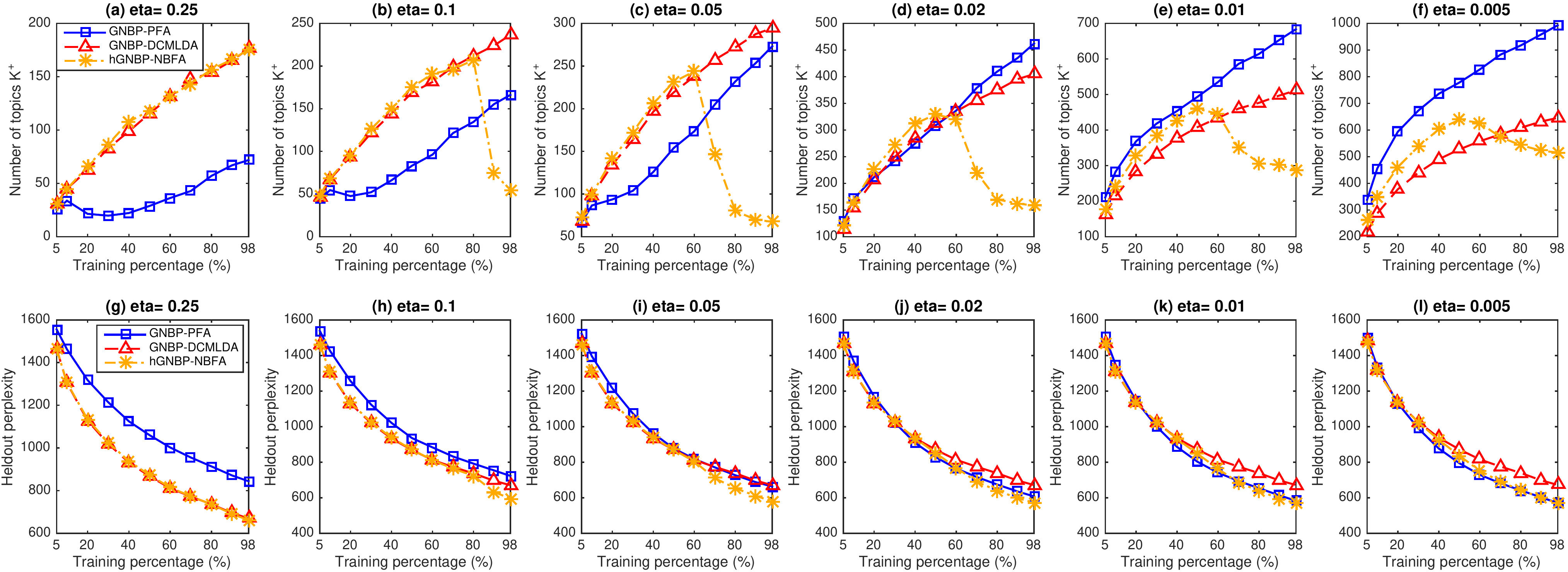}
\end{center}
\vspace{-6.0mm}
\caption{\small \label{fig:perplexity3_2}
Analogous plots to those in Figure \ref{fig:perplexity1_2} for the PsyReview dataset. 
}
\end{figure}

\begin{figure}[!h]
\begin{center}
\includegraphics[width=1\columnwidth]{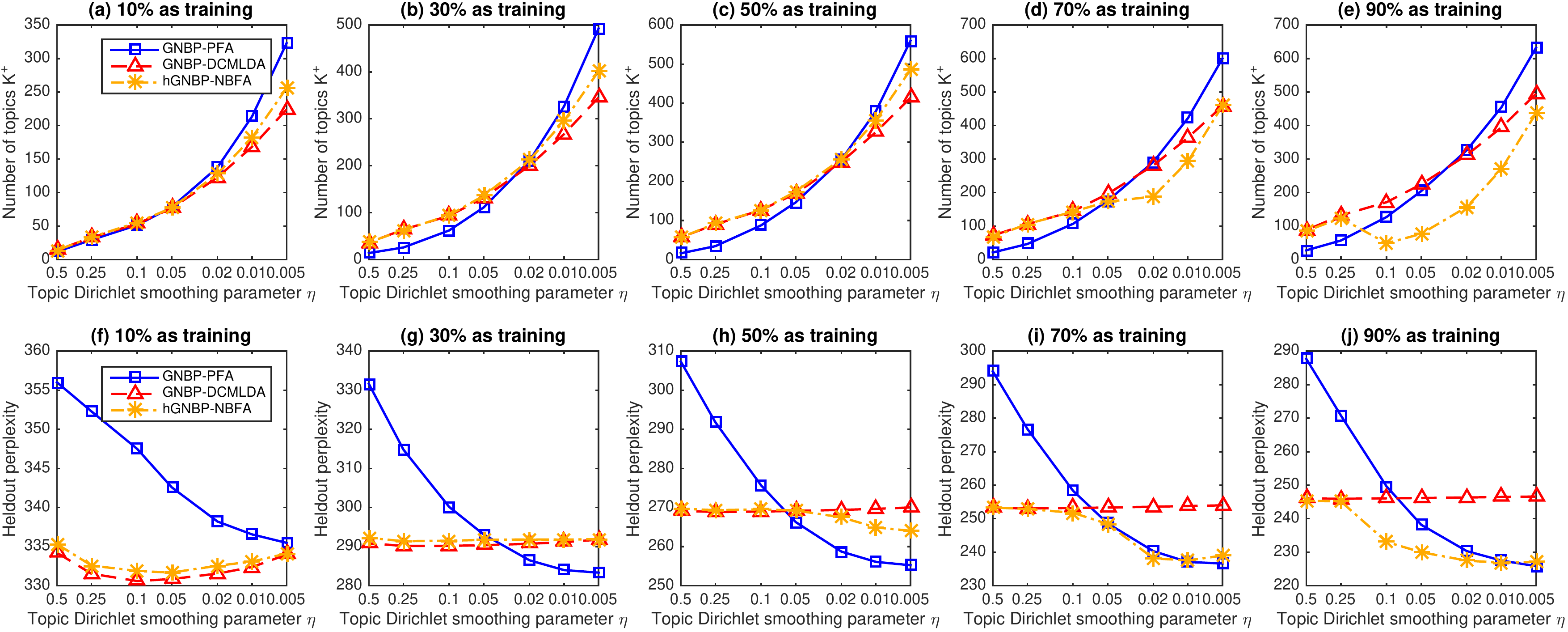}
\end{center}
\vspace{-6.0mm}
\caption{\small \label{fig:perplexity4_1}
Analogous plots to those in Figure \ref{fig:perplexity1_1} for the JACM dataset. 
}
\vspace{5mm}
\begin{center}
\includegraphics[width=1\columnwidth]{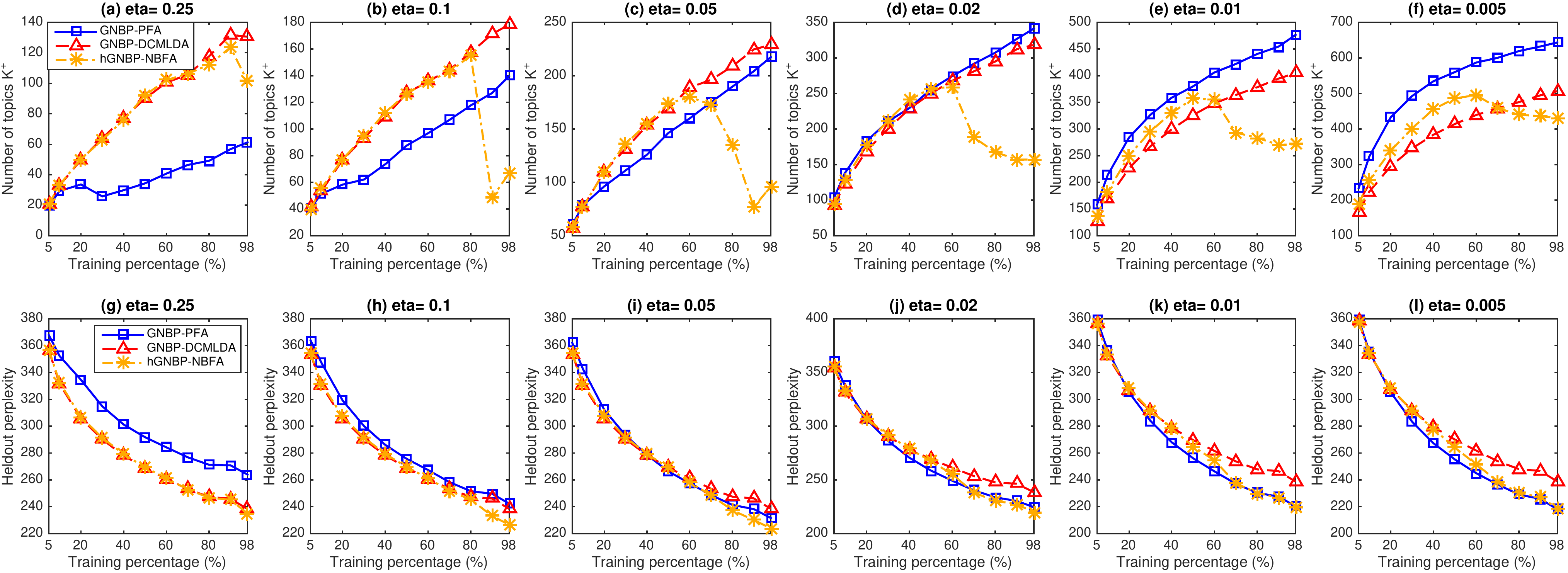}
\end{center}
\vspace{-6.0mm}
\caption{\small \label{fig:perplexity4_2}
Analogous plots to those in Figure \ref{fig:perplexity1_2} for the JACM dataset. 
}
\end{figure}

\end{document}